\documentclass[11pt, letterpaper]{article}

\usepackage{RR}
\RRNo{6391}
\usepackage[latin1]{inputenc} 
\usepackage[english]{babel}

\usepackage{amsmath}
\usepackage{units}
\usepackage{graphics}
\usepackage{epsfig}
\usepackage{latexsym}
\usepackage{amssymb}
\usepackage{color}
\usepackage{url}
\usepackage{geometry}
\usepackage{comment}

\graphicspath{{RR-Inria-4.10/logos/}}

 \newtheorem{theorem}{Theorem}[section]

 \newtheorem{lemma}[theorem]{Lemma}
 
 \newtheorem{corollary}[theorem]{Corollary}

 {\refstepcounter{figure}}
 {\smallskip}

 \newenvironment{pliste}%
 { \begin{list}%
         {--}%
         {\setlength{\labelwidth}{20pt}%
          \setlength{\leftmargin}{25pt}%
          \setlength{\topsep}{0pt}
          \setlength{\itemsep}{0pt}
          \setlength{\parsep}{0pt}}}%
 { \end{list} }

 \newenvironment{remark*}
         {\bigskip\noindent \textbf{Remark.} \hspace{0.3mm}} 
         {\medskip} 

 \newenvironment{remarks*}
         {\bigskip\noindent \textbf{Remarks.} \hspace{0.3mm}} 
         {\medskip} 

 \newcommand{\qed}{$\square$}

 \newenvironment{proof}
         {\smallskip\noindent \textbf{Proof.} \hspace{0.3mm}} 
         {\qed  \medskip} 

 \newenvironment{proofarg}[1]
         {\smallskip\noindent \textbf{Proof #1.}} 
         {\qed  \medskip}

 \newcommand{\sind}{\mbox{}\hspace{0.5cm}}

 \newcommand{\im}{{\rm im}\;}
 \newcommand{\rank}{{\rm rank}\;}
 \newcommand{\isom}{\cong}
 
 \newcommand{\id}{{\rm id}}

 \newcommand{\e}{\varepsilon}

 \newcommand{\sss}{{X}}

 \newcommand{\tangent}{T}
 \newcommand{\tangentS}{\tangent}

 \newcommand{\dual}[1]{\textrm{V}(#1)} 
 \newcommand{\wdual}[1]{\weighted{\textrm{V}}(#1)}

 \newcommand{\volume}{{\rm vol}} 
 \newcommand{\conv}{\textrm{CH}}

 \newcommand{\proj}{p}
 
 \newcommand{\projS}{\proj_\sss} 
  
 \newcommand{\dist}{{\rm d}} 

 \newcommand{\disth}{\dist_{\cal H}}

 \newcommand{\R}{\mathbb{R}} 
  
 \newcommand{\N}{\mathbb{N}} 

 \newcommand{\diam}{{\rm diam}} 
 \newcommand{\diamcc}{\diam_{\rm CC}}

 \newcommand{\ma}{{\rm M}}

 \newcommand{\rch}{{\rm rch}}

 \newcommand{\vor}{{\cal V}} 
 \newcommand{\del}{{\cal D}} 
 \newcommand{\wvor}{\weighted\vor} 
 \newcommand{\wdel}{\weighted\del}

 \newcommand{\vorsarg}[1]{\vor^{#1}} 
  
 \newcommand{\wvorsarg}[1]{\wvor^{#1}}
 \newcommand{\wvors}{\wvorsarg{\sss}}
 \newcommand{\delsarg}[1]{\del^{#1}} 
 \newcommand{\dels}{\delsarg{\sss}} 
 \newcommand{\wdelsarg}[1]{\wdel^{#1}}
 \newcommand{\wdels}{\wdelsarg{\sss}}

 \newcommand{\wfs}{{\rm wfs}}

 \newcommand{\tab }{\hspace{0.25cm} \begin{minipage}{7cm}} 
 \newcommand{\fintab }{\end{minipage} \par} 
 \newcommand{\algo}{\begin{center}\parbox{16cm}} 
 \newcommand{\finalgo}{\end{center}}

 \newcommand{\simplex}{\sigma}

 \newcommand{\complex}{{\cal C}} 

 \newcommand{\cech}[1]{\complex^{#1}}
 \newcommand{\cecha}{\cech{\alpha}}
 \newcommand{\cechap}{\cech{\alpha'}}

 \newcommand{\rips}[1]{{\cal R}^{#1}}
 \newcommand{\ripsa}{\rips{\alpha}}
 \newcommand{\ripsap}{\rips{\alpha'}}
 \newcommand{\ripsqe}{\rips{4\e}}
 \newcommand{\ripsse}{\rips{16\e}}

 \newcommand{\wit}{W}
 \newcommand{\eee}{{L}}
 \newcommand{\winf}[1]{\complex_{#1}}
 \newcommand{\winfw}{\winf{\wit}}
 \newcommand{\winfr}[2]{\winf{#2}^{#1}}
 \newcommand{\winfrw}[1]{\winfr{#1}{\wit}}
 \newcommand{\winfa}[1]{\winfr{\alpha}{#1}}
 \newcommand{\winfaw}{\winfa{\wit}}
 \newcommand{\winfap}[1]{\winfr{\alpha'}{#1}}
 \newcommand{\winfapw}{\winfap{\wit}}

 \newcommand{\eeei}{\eee(i)}
 \newcommand{\ei}{\e(i)}

 \newcommand{\io}{i_0}
 \newcommand{\eeeio}{\eee(i_0)}
 \newcommand{\eio}{\e(i_0)}

 \newcommand{\ripsqi}{\rips{4\ei}}
 \newcommand{\ripssi}{\rips{16\ei}}

 \newcommand{\weight}{\omega}
 
 \newcommand{\weighted}[1]{#1_{\weight}}

 \newcommand{\wpb}{{\bar \weight}}  

\RRdate{December 2007}

\RRauthor{ Frédéric Chazal\thanks[address]{ INRIA Futurs, Parc Orsay
    Université, 4, rue Jacques Monod - Bât. P, 91893 ORSAY Cedex,
    France. {\tt \{frederic.chazal, steve.oudot\}@inria.fr}
} \and Steve
  Y. Oudot\thanksref{address}}

\authorhead{Chazal \& Oudot}

\RRetitle{Towards Persistence-Based Reconstruction in Euclidean Spaces} 
\RRtitle{Vers une reconstruction bas\'ee sur la persistance dans les espaces euclidiens} 
\titlehead{Towards Persistence-Based Reconstruction in Euclidean Spaces}

\RRresume{La reconstruction de vari\'et\'es a \'et\'e fortement \'etudi\'ee
durant cette derni\`ere d\'ecennie, en particulier dans le cas des
petites dimensions. Des avanc\'ees r\'ecentes dans le cas des plus
grandes dimensions ont permis l'\'emergence de nouvelles m\'ethodes de
reconstruction qui peuvent traiter des nuages de points issus de
sous-vari\'et\'es lisses de $\R^d$ de dimensions
arbitraires. Toutefois, la complexit\'e de ces approches cro\^it
exponentiellement avec la dimension $d$ de l'espace ambiant, ce qui les
rend impraticables en dimensions moyennes ou grandes, meme pour
reconstruire des sous-vari\'et\'es de petite dimension telles que des
courbes ou des surfaces.

Dans cet article, nous introduison une nouvelle approche qui se situe
\`a la fronti\`ere entre la reconstruction classique et l'inf\'erence
topologique, et dont la complexit\'e cro\^it avec la dimension
intrins\`eque des donn\'ees. Notre algorithme combine deux paradigmes
: le raffinement glouton type maxmin et la persistence
topologique. Plus pr\'ecis\'ement, \'etant donn\'e un nuage de points
dans $\R^d$, l'algorithme construit un sous-ensemble de landmarks
it\'erativement, tout en maintenant une paire de complexes simpliciaux
imbriqu\'es, dont les images dans $\R^d$ sont proches des donn\'ees,
et dont l'homologie persistante coincide avec l'homologie de l'espace
sous-jacent aux donn\'ees. Quand le nuage de point est suffisamment
dens\'ement \'echantillonn\'e \`a partir d'une sous-vari\'et\'e lisse
de $\R^d$, notre m\'ethode retrouve l'homologie de la vari\'et\'e en
temps $c(m)n^5$, o\`u $n$ est la taille de l'entr\'ee et $c(m)$ est
une constante d\'ependant uniquement de la dimension intrins\`eque $m$
de la vari\'et\'e. Notre approche peut aussi reconstruire avec
garanties une large classe d'objets compacts dans $\R^d$, avec de
moins bons temps de calcul toutefois.

Afin de donner des garanties th\'eoriques \`a notre algorithme, nous
\'etudions les filtrations de \v Cech, de Rips, et de complexes de
t\'emoins dans $\R^d$, pour lesquels nous pr\'esentons un ensemble de
r\'esultats nouveaux. Plus pr\'ecis\'ement, nous montrons comment des
r\'esultats existants sur les unions de boules peuvent \^etre
transf\'er\'es aux filtrations de \v Cech, puis de l\`a aux
filtrations de Rips et de complexes de t\'emoins. Nous proposons
\'egalement une premi\`ere quantification d'une conjecture de Carlsson
et de Silva, selon laquelle les filtrations de complexes de t\'emoins
fournissent de meilleurs r\'esultats que les filtrations de \v Cech et
de Rips dans le cadre de l'inf\'erence topologique, en tout cas pour
le cas des sous-vari\'et\'es lisses de $\R^d$.
}

\RRmotcle{Reconstruction, Homologie persistante, Filtration, Complexe de \v Cech, Complexe de Rips, Complex de témoins, Inférence topologique}

\RRabstract{
Manifold reconstruction has been extensively studied among the
computational geometry community for the last decade or so, especially
in two and three dimensions. Recently, significant improvements were
made in higher dimensions, leading to new methods to
reconstruct large classes of compact subsets of Euclidean space
$\R^d$. However, the complexities of these methods scale up
exponentially with d, which makes them impractical in medium or high
dimensions, even for handling low-dimensional submanifolds.

In this paper, we introduce a novel approach that stands in-between
reconstruction and topological estimation, and whose complexity scales
up with the intrinsic dimension of the data. Our algorithm combines
two paradigms: greedy refinement, and topological
persistence. Specifically, given a point cloud in $\R^d$, the
algorithm builds a set of landmarks iteratively, while maintaining
nested pairs of complexes, whose images in $\R^d$ lie close to the
data, and whose persistent homology eventually coincides with the one
of the underlying shape. When the data points are sufficiently densely
sampled from a smooth $m$-submanifold of $\R^d$, our method retrieves
the homology of the submanifold in time at most $c(m)n^5$, where $n$
is the size of the input and $c(m)$ is a constant depending solely on
$m$. It can also provably well handle a wide range of compact subsets
of $\R^d$, though with worse complexities.

Along the way to proving the correctness of our algorithm, we obtain
new results on \v Cech, Rips, and witness complex filtrations in
Euclidean spaces. Specifically, we show how previous results on unions
of balls can be transposed to \v Cech filtrations. Moreover, we
propose a simple framework for studying the properties of filtrations
that are intertwined with the \v Cech filtration, among which are the
Rips and witness complex filtrations.  Finally, we investigate further
on witness complexes and quantify a conjecture of Carlsson and de
Silva, which states that witness complex filtrations should have
cleaner persistence barcodes than \v Cech or Rips filtrations, at
least on smooth submanifolds of Euclidean spaces.
}

\RRkeyword{Reconstruction, Persistent Homology, Filtration, \v Cech complex, Rips complex, Witness complex, Topological estimation}

\RRprojet{G\'eometrica}

\RRtheme{\THSym}

\RCSaclay

\begin{document}

\makeRR

%
%

%
\section{Introduction}
\label{sec-intro}

The problem of reconstructing unknown structures from finite
collections of data samples is ubiquitous in the Sciences, where it
has many different variants, depending on the nature of the data and
on the targeted application. In the last decade or so, the
computational geometry community has gained a lot of interest in
manifold reconstruction, where the goal is to reconstruct submanifolds
of Euclidean spaces from point clouds. In particular, efficient
solutions have been proposed in dimensions two and three, based on the
use of the Delaunay triangulation -- see \cite{cg-dtbsr-07} for a
survey. In these methods, the unknown manifold is approximated by a
simplicial complex that is extracted from the full-dimensional
Delaunay triangulation of the input point cloud.  The success of this
approach is explained by the fact that, not only does it behave well
on practical examples, but the quality of its output is guaranteed by
a sound theoretical framework. Indeed, the extracted complex is
usually shown to be equal, or at least close, to the so-called {\em
  restricted Delaunay triangulation}, a particular subset of the
Delaunay triangulation whose approximation power is well-understood on
smooth or Lipschitz curves and surfaces \cite{ab-srvf-99,
  abe-cbscc-98, bo-lrk-06}.  Unfortunately, the size of the Delaunay
triangulation grows too fast with the dimension of the ambient space
for the approach to be still tractable in high-dimensional spaces
\cite{m-mnfcp-70}.

Recently, significant steps were made towards a full understanding of
the potential and limitations of the restricted Delaunay triangulation
on smooth manifolds \cite{cdr-mrps-2005, o-ntrdwchd-07}. In parallel,
new sampling theories were developped, such as the critical point
theory for distance functions \cite{ccl-pc-06}, which provides
sufficient conditions for the topology of a shape $X\subset\R^d$ to be
captured by the offsets of a point cloud $\eee$ lying at small
Hausdorff distance. These advances lay the foundations of a new
theoretical framework for the reconstruction of smooth submanifolds
\cite{cl-tgmr-06, nsw-fhswhc-04}, and more generally of large classes
of compact subsets of $\R^d$ \cite{ccl-pc-06, cl-lma-05,
  cl-sctis-07}. Combined with the introduction of more lightweight
data structures, such as the {\em witness complex} \cite{ds-wdd-03},
they have lead to new reconstruction techniques in arbitrary Euclidean
spaces \cite{bgo-mradwc-07}, whose outputs can be guaranteed under
mild sampling conditions, and whose complexities can be orders of
magnitude below the one of the classical Delaunay-based approach.  For
instance, on a data set with $n$ points in $\R^d$, the algorithm of
\cite{bgo-mradwc-07} runs in time $2^{O(d^2)}n^2$, whereas the size of
the Delaunay triangulation can be of the order of
$n^{\left\lceil\frac{d}{2}\right\rceil}$. Unfortunately,
$2^{O(d^2)}n^2$ still remains too large for these new methods to be
practical, even when the data points lie on or near a very
low-dimensional submanifold.

A weaker yet similarly difficult version of the reconstruction
paradigm is topological estimation, where the goal is not to exhibit a
data structure that faithfully approximates the underlying shape $X$,
but simply to infer the topological invariants of $X$ from an input
point cloud $\eee$. This problem has received a lot of attention in
the recent years, and it finds applications in a number of areas of
Science, such as sensor networks \cite{dg-csnph-07}, statistical
analysis \cite{cisz-lbsni-07}, or dynamical systems \cite{kmm-ch-04,
  r-tcha-99}.  A classical approach to learning the homology of $X$
consists in building a nested sequence of spaces ${\cal
  K}^0\subseteq{\cal K}^1\subseteq\cdots\subseteq{\cal K}^m$, and in
studying the persistence of homology classes throughout this
sequence. In particular, it has been independently proved in
\cite{cl-sctis-07} and \cite{ceh-spd-05} that the persistent homology
of the sequence defined by the $\alpha$-offsets of a point cloud
$\eee$ coincides with the homology of the underlying shape $X$, under
sampling conditions that are milder than the ones of
\cite{ccl-pc-06}. Specifically, if the Hausdorff distance between
$\eee$ and $X$ is less than $\e$, for some small enough $\e$, then,
for all $\alpha\geq\e$, the canonical inclusion map
$\eee^\alpha\hookrightarrow\eee^{\alpha+2\e}$ induces homomorphisms
between homology groups, whose images are isomorphic to the homology
groups of $X$. Combined with the structure theorem of
\cite{zc-cph-05}, which states that the persistent homology of the
sequence $\{\eee^\alpha\}_{\alpha\geq 0}$ is fully described by a
finite set of intervals, called a {\em persistence barcode} or a {\em
  persistence diagram} --- see Figure \ref{fig-spiral} (left), the
above result means that the homology of $X$ can be deduced from this
barcode, simply by removing the intervals of length less than $2\e$,
which are therefore viewed as topological noise.

From an algorithmic point of view, the persistent homology of a nested
sequence of simplicial complexes (called a {\em filtration}) can be
efficiently computed using the persistence algorithm \cite{elz-tps-02,
  zc-cph-05}. Among the many filtrations that can be built on top of a
point set $\eee$, the $\alpha$-shape enables to reliably recover the
homology of the underlying space $X$, since it is known to be a
deformation retract of $\eee^\alpha$ \cite{e-ubids-95}. However, this
property is useless in high dimensions, since computing the
$\alpha$-shape requires to build the full-dimensional Delaunay
triangulation. It is therefore appealing to consider other filtrations
that are easy to compute in arbitrary dimensions, such as the Rips and
witness complex filtrations. Nevertheless, to the best of our
knowledge, there currently exists no equivalent of the result of
\cite{cl-sctis-07, ceh-spd-05} for such
filtrations. In this paper, we produce such a result, not only for
Rips and witness complexes, but more generally for any
filtration that is intertwined with the \v Cech filtration. Recall
that, for all $\alpha>0$, the \v Cech complex $\cecha(\eee)$ is the
nerve of the union of the open balls of same radius $\alpha$ about the
points of $\eee$, {\em i.e.}  the nerve of $\eee^\alpha$. It follows
from the nerve theorem \cite[Cor. 4G.3]{h-at-01} that $\cecha(\eee)$
and $\eee^\alpha$ are homotopy equivalent. However, despite the result
of \cite{cl-sctis-07, ceh-spd-05}, this is not sufficient to prove
that the persistent homology of
$\cecha(\eee)\hookrightarrow\cech{\alpha+2\e}(\eee)$ coincides with
the homology of $X$,
mainly because it is not clear whether the homotopy equivalences
$\cecha(\eee)\rightarrow\eee^\alpha$ and
$\cech{\alpha+2\e}(\eee)\rightarrow\eee^{\alpha+2\e}$ provided by the
nerve theorem commute with the canonical inclusions
$\cecha(\eee)\hookrightarrow\cech{\alpha+2\e}(\eee)$ and
$\eee^\alpha\hookrightarrow\eee^{\alpha+2\e}$. Using standard
arguments of algebraic topology, we prove that there exist some
homotopy equivalences
that do commute with the canonical inclusions, at least at homology
and homotopy levels. This enables us to extend the result of
\cite{cl-sctis-07, ceh-spd-05} to the \v Cech filtration, and 
from there to the Rips and witness complex filtrations.

\begin{figure}[htb]
\centering
\includegraphics[scale=1]{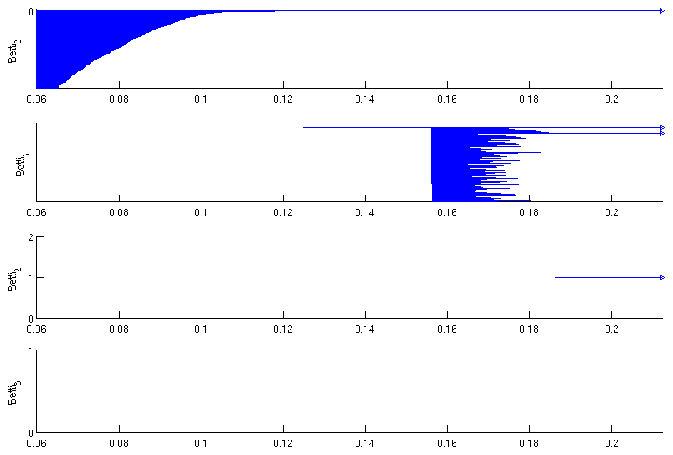}
\hspace{0.01\linewidth}
\includegraphics[scale=1]{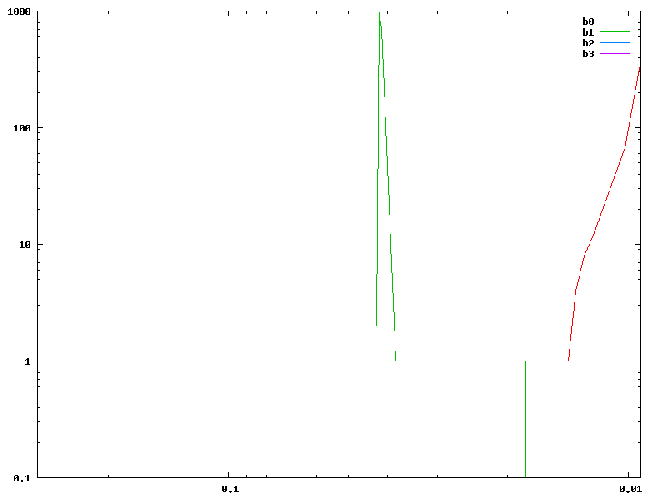}
\caption{Results obtained from a set $\wit$ of $10,000$
  points sampled uniformly at random from a helical curve drawn on the
  2d torus $(u,v)\mapsto \frac{1}{2}(\cos 2\pi u, \sin 2\pi u, \cos
  2\pi v, \sin 2\pi v)$ in $\R^4$ --- see \cite{go-ruwc-07}. Left:
  persistence barcode of the Rips filtration, built over a set of
  $900$ carefully-chosen landmarks. Right: result of our algorithm,
  applied blindly to the input $\wit$. Both methods highlight
  the two underlying structures: curve and torus.}
\label{fig-spiral}
\end{figure}

Another common concern in topological data analysis is the size of the
vertex set on top of which a filtration is built. In many practical
situations indeed, the point cloud $\wit$ given as input samples the
underlying shape very finely. In such situations, it makes sense to
build the filtration on top of a small subset $\eee$ of landmarks, to
avoid a waste of computational resources. However, building a
filtration on top of the sparse landmark set $\eee$ instead of the
dense point cloud $\wit$ can result in a significant degradation in
the quality of the persistence barcode. This is true in particular
with the \v Cech and Rips filtrations, whose barcodes can have
topological noise of amplitude depending directly on the density of
$\eee$. The introduction of the witness complex filtration appeared as
an elengant way of solving this issue \cite{cds-teuwc-04}. The witness
complex of $\eee$ relative to $\wit$, or $\winfw(\eee)$ for short, can
be viewed as a relaxed version of the Delaunay triangulation of
$\eee$, in which the points of $\wit\setminus\eee$ are used to drive
the construction of the complex \cite{ds-wdd-03}. Due to its special
nature, which takes advantage of the points of $\wit\setminus\eee$,
and due to its close relationship with the restricted Delaunay
triangulation, the witness complex filtration is likely to give
persistence barcodes whose topological noise depends on the density of
$\wit$ rather than on the one of $\eee$, as conjectured in
\cite{cds-teuwc-04}.  We prove in the paper that this statement is
only true to some extent, namely: whenever the points of $\wit$ are
sufficiently densely sampled from some smooth submanifold of $\R^d$,
the topological noise in the barcode can be arbitrarily small compared
to the density of $\eee$. Nevertheless, it cannot depend solely on the
density of $\wit$. This shows that the witness complex filtration does
provide cleaner persistence barcodes than \v Cech or Rips filtrations,
but maybe not as clean as expected.

Taking advantage of the above theoretical results on Rips and witness
complexes, we propose a novel approach to reconstruction that stands
somewhere in-between the classical reconstruction and topological
estimation paradigms. Our algorithm is a variant of the method of
\cite{bgo-mradwc-07, go-ruwc-07} that combines greedy refinement and
topological persistence. Specifically, given an input point cloud
$\wit$, the algorithm builds a subset $\eee$ of landmarks iteratively,
and in the meantime it maintains a nested pair of simplicial complexes
(which happen to be Rips or witness complexes) and computes its
persistent Betti numbers. The outcome of the algorithm is the sequence
of nested pairs maintained throughout the process, or rather the
diagram of evolution of their persistent Betti numbers. Using this
diagram, a user or software agent can determine a relevant scale at
which to process the data. It is then easy to rebuild the
corresponding set of landmarks, as well as its nested pair of
complexes. 
Note that our method does not completely solve the classical
reconstruction problem, since it does not exhibit an embedded complex
that is close to $X$ topologically and geometrically. Nevertheless, it
comes with theoretical guarantees, it is easily implementable, and
above all it has reasonable complexity. Indeed, in the case where the
input point cloud is sampled from a smooth submanifold $\sss$ of
$\R^d$, we show that the complexity of our algorithm is bounded by
$c(m)n^5$, where $c(m)$ is a quantity depending solely on the
intrinsic dimension $m$ of $X$, while $n$ is the size of the input. To
the best of our knowledge, this is the first provably-good topological
estimation or reconstruction method whose complexity scales up with
the intrinsic dimension of the manifold. In the case where $X$ is a
more general compact set in $\R^d$, our complexity bound becomes
$c(d)n^5$.

The paper is organized as follows: after introducing the \v Cech,
Rips, and witness complex filtrations in Section
\ref{sec-def-complexes}, we prove our structural results in Sections
\ref{sec-struct-res-compact} and \ref{sec-smooth-case}, focusing on
the general case of compact subsets of $\R^d$
in Section \ref{sec-struct-res-compact},
 and more specifically on the case of smooth submanifolds of $\R^d$
in Section \ref{sec-smooth-case}. Finally, we present our algorithm
and its analysis in Section \ref{sec-recons}.

%


\section{Various complexes and their relationships}
\label{sec-def-complexes}

The definitions, results and proofs of this section hold in any
arbitrary metric space. However, for the sake of consistency with the
rest of the paper, we state them in the particular case of $\R^d$,
endowed with the Euclidean norm $\|p\|=\sqrt{\sum_{i=1}^{d}
  p_i^2}$. As a consequence, our bounds are not the tightest possible
for the Euclidean case, but they are for the general metric case.
Using specific properties of Euclidean spaces, it is indeed possible
to work out somewhat tighter bounds, but at the price of a loss of
simplicity in the statements.

For any compact set $X\subset\R^d$, we call $\diam(X)$ the
diameter of $X$, and $\diamcc(X)$ the {\em component-wise diameter} of
$X$, defined by: $\diamcc(X)=\inf_i \diam(X_i)$, where the
$X_i$ are the path-connected components of
$X$. Finally, given two compact sets $X,Y$ in $\R^d$, we call
$\disth(X,Y)$ their Hausdorff distance.

\paragraph{\v Cech complex.}
Given a finite set $\eee$ of points of $\R^d$ and a positive number
$\alpha$, we call $\eee^\alpha$ the union of the open balls of radius
$\alpha$ centered at the points of $\eee$: $\eee^\alpha = \bigcup_{x\in \eee}
B(x,\alpha)$. This definition makes sense only for $\alpha>0$, since
for $\alpha=0$ we get $\eee^\alpha=\emptyset$. We also denote by
$\{\eee^\alpha\}$ the open cover of $\eee^\alpha$ formed by the open balls
of radius $\alpha$ centered at the points of $\eee$. The \v Cech complex
of $\eee$ of parameter $\alpha$, or $\cecha(\eee)$ for short, is the {\em
  nerve} of this cover, {\em i.e.} it is the abstract simplicial
complex whose vertex set is $\eee$, and such that, for all $k\in\N$ and
all $x_0,\cdots,x_k\in \eee$, $[x_0,\cdots,x_k]$ is a $k$-simplex of
$\cecha(\eee)$ if and only if $B(x_0,\alpha)\cap\cdots\cap
B(x_k,\alpha)\neq\emptyset$.

\paragraph{Rips complex.}
Given a finite set $\eee\subset\R^d$ and a positive number $\alpha$,
the Rips complex of $\eee$ of parameter $\alpha$, or $\ripsa(\eee)$ for
short, is the abstract simplicial complex whose $k$-simplices
correspond to unordered $(k+1)$-tuples of points of $\eee$ which are
pairwise within Euclidean distance $\alpha$ of one another. The Rips
complex is closely related to the \v Cech complex, as stated in the
following standard lemma, whose proof is recalled for completeness:
\begin{lemma}\label{lem-cech-vs-rips}
For all finite set $\eee\subset\R^d$ and all $\alpha>0$, we have:
$\cech{\frac{\alpha}{2}}(\eee)\subseteq\ripsa(\eee)\subseteq\cecha(\eee)$.
\end{lemma}
\begin{proof}
The proof is standard. Let $[x_0,\cdots, x_k]$ be an arbitrary
$k$-simplex of $\cech{\frac{\alpha}{2}}(\eee)$. The Euclidean balls of
same radius $\frac{\alpha}{2}$ centered at the $x_i$ have a non-empty
common intersection in $\R^d$. Let $p$ be a point in the
intersection. We then have: $\forall 0\leq i,j\leq k$,
$\|x_i-x_j\|\leq\|x_i-p\|+\|p-x_j\|\leq \alpha$. This implies that
$[x_0,\cdots,x_k]$ is a simplex of $\ripsa(\eee)$, which proves the
first inclusion of the lemma.

Let now $[x_0,\cdots, x_k]$ be an arbitrary $k$-simplex of
$\ripsa(\eee)$. We have $\|x_0-x_i\|\leq\alpha$ for all $i=0, \cdots,
k$. This means that $x_0$ belongs to all the Euclidean balls
$B(x_i,\alpha)$, which therefore have a non-empty common intersection
in $\R^d$. It follows that $[x_0,\cdots,x_k]$ is a simplex of
$\cecha(\eee)$, which proves the second inclusion of the lemma.
\end{proof}

\paragraph{Witness complex.}
Let $\eee$ be a finite subset of $\R^d$, referred to as the landmark
set, and let $\wit$ be another (possibly infinite) subset of $\R^d$,
identified as the witness set. Let also $\alpha\in [0,\infty)$.
\begin{pliste}
\item Given a point
$w\in\wit$ and a $k$-simplex $\simplex$ with vertices in $\eee$, $w$ is an
{\em $\alpha$-witness} of $\simplex$ (or, equivalently, $w$ {\em
  $\alpha$-witnesses} $\simplex$) if the vertices of $\simplex$ lie
within Euclidean distance $(\dist_k(w)+\alpha)$ of $w$, where
$\dist_k(w)$ denotes the Euclidean distance between $w$ and its
$(k+1)$th nearest landmark in the Euclidean metric. 
\item The {\em $\alpha$-witness complex} of $\eee$ {\em
  relative to} $\wit$, or $\winfaw(\eee)$ for short, is the maximum
  abstract simplicial complex, with vertices in $\eee$, whose faces
  are $\alpha$-witnessed by points of $\wit$.
\end{pliste}
When $\alpha=0$, the $\alpha$-witness
complex coincides with the standard witness complex $\winfw(\eee)$,
introduced in \cite{ds-wcd-07}.  The $\alpha$-witness
complex is also closely related to the \v Cech complex, though the
relationship is a bit more subtle than in the case of the Rips
complex:
\begin{lemma}\label{lem-cech-vs-winf}
Let $\eee,\wit\subseteq\R^d$ be such that $\eee$ is finite. If every
point of $\eee$ lies within Euclidean distance $l$ of $\wit$, then for
all $\alpha>l$ we have:
$\cech{\frac{\alpha-l}{2}}(\eee)\subseteq\winfaw(\eee)$. In addition,
if the Euclidean distance from any point of $\wit$ to its second
nearest neighbor in $\eee$ is at most $l'$, then for all $\alpha>0$ we
have: $\winfaw(\eee)\subseteq\cech{2(\alpha+l')}(\eee)$.
\end{lemma}
\begin{proof}
Let $[x_0,\cdots,x_k]$ be a $k$-simplex of
$\cech{\frac{\alpha-l}{2}}(\eee)$. This means that $\bigcap_{i=0}^k
B(x_i,\frac{\alpha-l}{2})\neq\emptyset$, and as a result, that
$\|x_0-x_i\|\leq \alpha-l$ for all $i=0,\cdots,k$. Let $w$ be a point
of $\wit$ closest to $x_0$ in the Euclidean metric. By the hypothesis
of the lemma, we have $\|w-x_0\|\leq l$, therefore $x_0,\cdots,x_k$
lie within Euclidean distance $\alpha$ of $w$. Since the Euclidean
distances from $w$ to its nearest points of $\eee$ are non-negative,
$w$ is an $\alpha$-witness of $[x_0,\cdots,x_k]$ and of all its
faces. As a result, $[x_0,\cdots,x_k]$ is a simplex of
$\winfaw(\eee)$.

Consider now a $k$-simplex $[x_0,\cdots, x_k]$ of $\winfaw(\eee)$. If
$k=0$, then the simplex is a vertex $[x_0]$, and therefore it belongs
to $\cechap(\eee)$ for all $\alpha'>0$. Assume now that $k\geq
1$. Edges $[x_0,x_1],\cdots, [x_0,x_k]$ belong also to
$\winfaw(\eee)$, hence they are $\alpha$-witnessed by points of
$\wit$. Let $w_i\in\wit$ be an $\alpha$-witness of
$[x_0,x_i]$. Distances $\|w_i-x_0\|$ and $\|w_i-x_i\|$ are bounded
from above by $\dist_2(w_i)+\alpha$, where $\dist_2(w_i)$ is the
Euclidean distance from $w_i$ to its second nearest point of $\eee$,
which by assumption is at most $l'$. It follows that
$\|x_0-x_i\|\leq\|x_0-w_i\|+\|w_i-x_i\|\leq 2\alpha+2l'$. Since this
is true for all $i=0,\cdots,k$, we conclude that $x_0$ belongs to the
intersection $\bigcap_{i=0}^k B(x_i, 2(\alpha+l'))$, which is
therefore non-empty. As a result, $[x_0,\cdots,x_k]$ is a simplex of
$\cech{2(\alpha+l')}(\eee)$.
\end{proof}
\begin{corollary}\label{cor-cech-vs-winf}
Let $X$ be a compact subset of $\R^d$, and let
$\eee\subseteq\wit\subseteq\R^d$ be such that $\eee$ is finite. Assume
that $\disth(X, \wit)\leq\delta$ and that $\disth(\wit, \eee)\leq\e$,
with $\e+\delta<\frac{1}{4}\;\diamcc(X)$. Then, for all
$\alpha>\e$, we have:
$\cech{\frac{\alpha-\e}{2}}(\eee)\subseteq\winfaw(\eee)
\subseteq\cech{2\alpha+6(\e+\delta)}(\eee)$. In particular, if
$\delta\leq\e<\frac{1}{8}\;\diamcc(X)$, then, for all $\alpha\geq
2\e$ we have:
$\cech{\frac{\alpha}{4}}(\eee)\subseteq\winfaw(\eee)\subseteq
\cech{8\alpha}(\eee)$.
\end{corollary}
\begin{proof}
Since $\disth(\wit, \eee)\leq\e$, every point of $\eee$ lies within
Euclidean distance $\e$ of $\wit$. As a result, the first inclusion of
Lemma \ref{lem-cech-vs-winf} holds with $l=\e$, that is:
$\cech{\frac{\alpha-\e}{2}}(\eee)\subseteq\winfaw(\eee)$. 

Now, for every point $w\in\wit$, there is a point $p\in\eee$ such that
$\|w-p\|\leq\e$. Moreover, there is a point $x\in X$ such that
$\|w-x\|\leq\delta$, since we assumed that $\disth(X,
\wit)\leq\delta$. Let $X_x$ be the path-connected component of $X$
that contains $x$. Take an arbitrary value $\lambda\in
\left(0,\frac{1}{2}\;\diamcc(X)-2(\e+\delta)\right)$, and consider
the open ball $B(w, 2(\e+\delta)+\lambda)$. This ball clearly
intersects $X_x$, since it contains $x$. Furthermore, $X_x$ is not
contained entirely in the ball, since otherwise we would have:
$\diamcc(X)\leq\diam(X_x)\leq 4(\e+\delta)+2\lambda$, hereby
contradicting the fact that
$\lambda<\frac{1}{2}\;\diamcc(X)-2(\e+\delta)$. Hence, there is a
point $y\in X$ lying on the bounding sphere of $B(w,
2(\e+\delta)+\lambda)$. Let $q\in\eee$ be closest to $y$. We have
$\|y-q\|\leq\e+\delta$, since our hypothesis implies that
$\disth(X,\eee)\leq\disth(X,\wit)+\disth(\wit,\eee)\leq\delta+\e$. It
follows then from the triangle inequality that $\|p-q\|\geq
\|w-y\|-\|w-p\|-\|y-q\|\geq
2(\e+\delta)+\lambda-(\e+\delta)-(\e+\delta)=\lambda>0$. Thus, $q$ is
different from $p$, and therefore the ball $B(w,
3(\e+\delta)+\lambda)$ contains at least two points of $\eee$. Since
this is true for arbitrarily small values of $\lambda$, the Euclidean
distance from $w$ to its second nearest neighbor in $\eee$ is at most
$3(\e+\delta)$.  It follows that the second inclusion of Lemma
\ref{lem-cech-vs-winf} holds with $l'=3(\e+\delta)$, that is:
$\winfaw(\eee)\subseteq\cech{2(\alpha+3(\e+\delta))}(\eee)$.
\end{proof}

As mentioned at the head of the section, slightly tighter bounds can
be worked out using specific properties of Euclidean spaces. For the
case of the Rips complex, this was done by de Silva and Ghrist
\cite{dg-csnph-07, g-bptd-07}. Their approach can be combined with
ours in the case of the witness complex.


\section{Structural properties of filtrations over compact subsets of $\R^d$}
\label{sec-struct-res-compact}

Throughout this section, we use classical concepts of algebraic
topology, such as homotopy equivalences, deformation retracts,
or singular homology. We refer the reader to \cite{h-at-01} for a good
introduction to these concepts.  

Given a compact set $X\subset\R^d$, we denote by $\dist_X$ the {\em
  distance function} defined by $\dist_X(x) = \inf \{\|x-y\| : y \in X
\}$. Although $\dist_X$ is not differentiable, it is possible to define a
notion of critical point for distance functions and we denote by
$\wfs(X)$ the {\em weak feature size} of $X$, defined as the smallest
positive critical value of the distance function to $X$
\cite{cl-lma-05}.  We do not explicitly use the notion of critical
value in the following, but only its relationship with the topology of
the {\em offsets} $X^\alpha = \{ x \in \R^d : \dist_X(x) \leq \alpha \}$,
stressed in the following result from \cite{g-cptdf-93}:

\begin{lemma}[Isotopy Lemma] \label{lemma:isotopy}
If $0 < \alpha < \alpha'$ are such that there is no critical value of
$\dist_X$ in the closed interval $[\alpha,\alpha']$, then $X^{\alpha}$
and $X^{\alpha'}$ are homeomorphic (and even isotopic), and
$X^{\alpha'}$ deformation retracts onto $X^{\alpha}$.
\end{lemma}

In particular the hypothesis of the lemma is satisfied when $0 <
\alpha_1 < \alpha_2 < \wfs(X)$. In other words, all the offsets of $X$
have the same topology in the interval $(0,\wfs(X))$.

\subsection{Results on homology}
\label{sec-res-homology}

We use singular homology with coefficients in an arbitrary field --
omitted in our notations. In the following, we repeatedly make use of
the following standard result of linear algebra:
\begin{lemma}[Sandwich Lemma]\label{sandwich-lemma}
Consider the following sequence of homomorphisms between
finite-dimensional vector spaces over a same field:
$
A\rightarrow
B\rightarrow
C\rightarrow
D\rightarrow
E\rightarrow
F.
$
Assume that $\rank (A \rightarrow F)=\rank(C \rightarrow D)$. Then,
this quantity also equals the rank of $B \rightarrow E$. In the same way, if 
$
A\rightarrow
B\rightarrow
C\rightarrow
E\rightarrow
F
$ 
is a sequence of homomorphisms such that $\rank (A \rightarrow F)=
\dim C$, then $\rank(B \rightarrow E) = \dim C$.
\end{lemma}
\begin{proof}
Observe that, for any sequence of homomorphisms
$F\stackrel{f}{\rightarrow}G\stackrel{g}{\rightarrow}H$, we have
$\rank (g\circ f)\leq\min\{\rank f, \rank g\}$. Applying this fact to
maps $A\rightarrow F$, $B\rightarrow E$, and $C\rightarrow D$, which
are nested in the sequence of the lemma, we get: $\rank (A\rightarrow
F)\leq\rank (B\rightarrow E)\leq\rank (C\rightarrow D)$, which proves
the first statement of the lemma. As for the second statement, it is
obtained from the first one by letting $D = C$ and taking
$C\rightarrow D$ to be the identity map.
\end{proof}

\subsubsection{\v Cech filtration}
\label{sec-prop-cech}

Since the \v Cech complex is the nerve of a union of balls, its
topological invariants can be read from the structure of its dual
union. It turns out that unions of balls have been extensively studied
in the past \cite{ccl-pc-06, cl-sctis-07, ceh-spd-05}. Our analysis
relies particularly on the following result, which is an easy
extension of Theorem 4.7 of \cite{cl-sctis-07}:
\begin{lemma} \label{lem-prop-union-of-balls}
Let $X$ be a compact set and $\eee$ a finite set in $\R^d$, such that
$\disth(X,\eee) < \e$ for some $\e < \frac{1}{4}\;\wfs(X)$.  Then, for
all $\alpha,\alpha'\in\left[\e, \wfs(X)-\e\right]$ such that
$\alpha'-\alpha \geq 2\e$, and for all $\lambda\in\left(0,\wfs(X)\right)$,
we have: $\forall k\in\N$, $H_k(X^\lambda)\isom \im i_*$, where
$i_*:H_k(\eee^\alpha)\rightarrow H_k(\eee^{\alpha'})$ is the
homomorphism between homology groups induced by the canonical
inclusion $i:\eee^\alpha\hookrightarrow\eee^{\alpha'}$.  Given an
arbitrary point $x_0 \in X$, the same conclusion holds for homotopy
groups with base-point $x_0$.
\end{lemma}
\begin{proof}
We can assume without loss of generality that $\e < \alpha < \alpha' -
2\e < \wfs(X) - 3\e$, since otherwise we can replace $\e$ by any $\e'
\in (d_H(X,\eee),\e)$. From the hypothesis
we deduce the following sequence of inclusions:
\begin{equation}\label{eq-sandwich-union-of-balls}
X^{\alpha - \e} \hookrightarrow \eee^\alpha
\hookrightarrow X^{\alpha + \e}
\hookrightarrow \eee^{\alpha'}
\hookrightarrow X^{{\alpha'} + \e}
\end{equation}
%
By the Isotopy Lemma \ref{lemma:isotopy}, for all
$0<\beta<\beta'<\wfs(X)$, the canonical inclusion
$X^\beta\hookrightarrow X^{\beta'}$ is a homotopy equivalence. As a
consequence, Eq. (\ref{eq-sandwich-union-of-balls}) induces a sequence
of homomorphisms between homology groups, such that all homomorphisms
between homology groups of $X^{\alpha-\e}, X^{\alpha+\e},
X^{\alpha'+\e}$ are isomorphisms. It follows then from the Sandwich
Lemma \ref{sandwich-lemma} that $i_*: H_k(\eee^\alpha)\rightarrow
H_k(\eee^{\alpha'})$ has same rank as these isomorphisms. Now, this
rank is equal to the dimension of $H_k(X^\lambda)$, since the
$X^\beta$ are homotopy equivalent to $X^\lambda$ for all
$0<\beta<\wfs(X)$. It follows that $\im i_*\isom\dim H_k(X^\lambda)$,
since our ring of coefficients is a field.  The case of homotopy
groups is a little trickier, since replacing homology groups by
homotopy groups does not allow us to use the above rank
argument. However, we can use the same proof as in Theorem 4.7 of
\cite{cl-sctis-07} to conclude.
\end{proof}

Observe that Lemma \ref{lem-prop-union-of-balls} does not guarantee
the retrieval of the homology of $X$. Instead, it deals with
sufficiently small offsets of $X$, which are homotopy equivalent to
one another but possibly not to $X$ itself. In the special case where
$X$ is a smooth submanifold
of $\R^d$ however, $X^\lambda$ and $X$ are homotopy equivalent, and
therefore the theorem guarantees the retrieval of the homology of $X$.
From an algorithmic point of view, the main drawback of Lemma
\ref{lem-prop-union-of-balls} is that computing the homology of a
union of balls or the image of the homomorphism $i_*$ is usually
awkward. As mentionned in \cite{cl-sctis-07,ceh-spd-05} this can be
done by computing the persistence of the $\alpha$-shape or
$\lambda$-medial axis filtrations associated to $L$ but there do not
exist efficient algorithms to compute these filtrations in dimension
more than $3$. In the following we show that we can still reliably
obtain the homology of $X$ from easier to compute filtrations, namely
the Rips and Witness complexes filtrations.

Consider now the \v Cech complex $\cecha(\eee)$, for any value
$\alpha>0$. By definition, $\cecha(\eee)$ is the nerve of the open
cover $\{\eee^\alpha\}$ of $\eee^\alpha$. Since the elements of
$\{\eee^\alpha\}$ are open Euclidean balls, they are convex, and
therefore their intersections are either empty or convex. It follows
that $\{\eee^\alpha\}$ satisfies the hypotheses of the {\em nerve
  theorem}, which implies that $\cecha(\eee)$ and $\eee^\alpha$ are
homotopy equivalent -- see {\em e.g.}  \cite[Corollary
  4G.3]{h-at-01}. 
We thus get the following diagram,
where horizontal arrows are canonical inclusions, and vertical arrows
are homotopy equivalences provided by the nerve theorem:
\begin{equation}\label{eq-cech-vs-union-of-balls}
\begin{array}{ccc}
\eee^\alpha&\hookrightarrow&\eee^{\alpha'}\\
\uparrow&~&\uparrow\\
\cecha(\eee)&\hookrightarrow&\cech{\alpha'}(\eee)
\end{array}
\end{equation}

Determining whether this diagram commutes is not straightforward. The
following result, based on standard arguments of algebraic topology,
shows that there exist homotopy equivalences between the union of
balls and the \v Cech complex that make the above diagram commutative
at homology and homotopy levels:
\begin{lemma} \label{lemma:parametric_nerve}
Let $L$ be a finite set of points in $\R^d$ and let $0 <\alpha <
\alpha'$. Then, there exist homotopy equivalences $\cecha(\eee)
\rightarrow L^\alpha$ and $\cechap(\eee)\rightarrow L^{\alpha'}$ 
such that, for all $k \in \N$, the diagram of Eq.
(\ref{eq-cech-vs-union-of-balls}) induces the following commutative diagrams:
\begin{equation*}
\begin{array}{ccccccc}
H_k(\eee^\alpha) & \rightarrow & H_k(\eee^{\alpha'}) & \hskip1cm &
\pi_k(\eee^\alpha) & \rightarrow &
\pi_k(\eee^{\alpha'})\\ \uparrow&~&\uparrow & \mbox{\rm and} &
\uparrow&~&\uparrow
\\ H_k(\cecha(\eee))&\rightarrow&H_k(\cech{\alpha'}(\eee)) & \hskip1cm
& \pi_k(\cecha(\eee)) & \rightarrow & \pi_k(\cech{\alpha'}(\eee))
\end{array}
\end{equation*}
where vertical arrows are isomorphisms.
\end{lemma}
\begin{proof}
Our approach consists in a quick review of the proof of the nerve
theorem provided in Section 4G of \cite{h-at-01}, and in a simple
extension of the main arguments to our context. 

As mentioned earlier, the open cover $\{\eee^\alpha\}$ satisfies the
conditions of the nerve theorem, namely: for all points $x_0,\cdots,
x_k\in\eee$, $\bigcap_{l=0}^k B(x_l, \alpha)$ is either empty, or
convex and therefore contractible. From this cover we construct a
topological space $\Delta L^\alpha$ as follows: let $\Delta^n$ denote
the standard $n$-simplex, where $n=\#\eee-1$. To each non-empty subset
$S$ of $\eee$ we associate the face $[S]$ of $\Delta^n$ spanned by the
elements of $S$, as well as the space $B_S(\alpha)=\bigcap_{s\in S}
B(s, \alpha)\subseteq\eee^\alpha$.  $\Delta L^\alpha$ is then the
subspace of $\eee^\alpha\times\Delta^n$ defined by:
$$
\Delta\eee^\alpha = \bigcup_{\emptyset\neq S\subseteq\eee} B_S(\alpha) \times [S]
$$ The space $\Delta L^{\alpha'}$ is built similarly. The product
structures of $\Delta\eee^\alpha$ and $\Delta L^{\alpha'}$ imply the
existence of canonical projections
$p_\alpha:\Delta\eee^\alpha\rightarrow \eee^\alpha$ and
$p_{\alpha'}:\Delta\eee^{\alpha'}\rightarrow \eee^{\alpha'}$. These
projections commute with the canonical inclusions $\Delta
L^\alpha\hookrightarrow\Delta L^{\alpha'}$ and
$L^\alpha\hookrightarrow L^{\alpha'}$, which implies that the
following diagram:
\begin{equation}\label{eq-blowup-diag1}
\begin{array}{ccc}
\eee^\alpha&\hookrightarrow&\eee^{\alpha'}\\
p_\alpha \uparrow&~&\uparrow p_{\alpha'}\\
\Delta\eee^\alpha&\hookrightarrow&\Delta\eee^{\alpha'}
\end{array}
\end{equation}
induces commutative diagrams at homology and homotopy levels.
Moreover, since $\{\eee^\alpha\}$ is an open cover of $\eee^\alpha$,
which is paracompact, $p_\alpha$ is a homotopy equivalence
\cite[Prop. 4G.2]{h-at-01}. The same holds for $p_{\alpha'}$, and
therefore $p_\alpha$ and $p_{\alpha'}$ induce isomorphisms at homology
and homotopy levels.

We now show that, similarly, there exist homotopy equivalences
$\Delta\eee^\alpha \rightarrow \cech{\alpha}(L)$ and
$\Delta\eee^{\alpha'} \rightarrow \cech{\alpha'}(L)$ that commute with
the canonical inclusions
$\Delta\eee^\alpha\hookrightarrow\Delta\eee^{\alpha'}$ and
$\cech{\alpha}(L)\hookrightarrow\cechap(\eee)$.  This follows in fact
from the proof of Corollary 4G.3 of \cite{h-at-01}. Indeed, using the
notion of {\em complex of spaces} introduced in \cite[Section
  4G]{h-at-01}, it can be shown that $\Delta L^\alpha$ is the
realization of the complex of spaces associated with the cover
$\{L^\alpha\}$ --- see the proof of \cite[Prop. 4G.2]{h-at-01}. Its
base is the barycentric subdivision $\Gamma^\alpha$ of
$\cech{\alpha}(L)$, where each vertex corresponds to a non-empty
finite intersection $B_S(\alpha)$ for some $S \subseteq L$, and where
each edge connecting two vertices $S\subset S'$ 
corresponds to the canonical
%
inclusion $B_{S'}(\alpha)\hookrightarrow B_S(\alpha)$.
In the same way, $\Delta\eee^{\alpha'}$ is the realization of a
complex of spaces built over the barycentric subdivision
$\Gamma^{\alpha'}$ of $\cech{\alpha'}(L)$. Now, since the non-empty
finite intersections $B_S(\alpha)$ (resp. $B_S(\alpha')$) are
contractible, the map $q_\alpha : \Delta L^\alpha \rightarrow
\Gamma^\alpha$ (resp. $q_{\alpha'} : \Delta\eee^{\alpha'}
\rightarrow \Gamma^{\alpha'}$) induced by sending each open set
$B_S(\alpha)$ (resp. $B_S(\alpha')$) to a point is a homotopy
equivalence \cite[Prop. 4G.1 and Corol. 4G.3]{h-at-01}.
Furthermore, by construction, $q_\alpha$ is the
restriction of $q_{\alpha'}$ to $\Delta L^\alpha$. Therefore,
\begin{equation}\label{eq-blowup-diag2}
\begin{array}{ccc}
\Delta\eee^\alpha&\hookrightarrow&\Delta\eee^{\alpha'}\\
q_\alpha \downarrow&~&\downarrow q_{\alpha'}\\
\Gamma^\alpha&\hookrightarrow&\Gamma^{\alpha'}\\
\end{array}
\end{equation}
is a commutative diagram where vertical arrows are homotopy
equivalences. Now, it is well-known that $\Gamma^\alpha$ and
$\Gamma^{\alpha'}$ are homeomorphic to $\cech{\alpha}(L)$ and
$\cech{\alpha'}(L)$ respectively, and that the homeomorphisms commute
with the inclusion. Combined with (\ref{eq-blowup-diag1}) and
(\ref{eq-blowup-diag2}), this fact proves Lemma~\ref{lemma:parametric_nerve}.
\end{proof}

Combining Lemmas \ref{lem-prop-union-of-balls} and
\ref{lemma:parametric_nerve}, we obtain the following key result:
\begin{theorem}\label{th-prop-cech}
Let $X$ be a compact set and $\eee$ a finite set in $\R^d$, such that
$\disth(X,\eee) < \e$ for some $\e < \frac{1}{4}\;\wfs(X)$.  Then, for
all $\alpha,\alpha'\in\left[\e, \wfs(X)-\e\right]$ such that
$\alpha'-\alpha>2\e$, and for all $\lambda\in\left(0,\wfs(X)\right)$,
we have: $\forall k\in\N$, $H_k(X^\lambda)\isom \im j_*$, where
$j_*:H_k(\cecha(\eee))\rightarrow H_k(\cech{\alpha'}(\eee))$ is the
homomorphism between homology groups induced by the canonical
inclusion $j:\cecha(\eee)\hookrightarrow\cech{\alpha'}(\eee)$. Given
an arbitrary point $x_0 \in X$, the same result holds for homotopy
groups with base-point $x_0$.
\end{theorem}
Using the terminology of \cite{zc-cph-05}, this result means that the
homology of $X^\lambda$ can be deduced from the persistent homology of
the filtration $\{\cecha(\eee)\}_{\alpha\geq 0}$ by removing the
cycles of persistence less than $2\e$. Equivalently, the amplitude of
the {\em topological noise} in the persistence barcode of
$\{\cecha(\eee)\}_{\alpha\geq 0}$ is bounded by $2\e$, {\em i.e.} the
intervals of length at least $2\e$ in the barcode give the homology of
$X^\lambda$.

\subsubsection{Filtrations intertwined with the \v Cech filtration}
\label{sec-prop-rips-winf}

Using Lemma \ref{lem-cech-vs-rips} and Theorem
\ref{th-prop-cech}, we get the following guarantees on the Rips
filtration:
\begin{theorem}\label{th-prop-rips}
Let $X\subset\R^d$ be a compact set, and $\eee\subset\R^d$ a finite
set such that $\disth(X,\eee)<\e$ for some
$\e<\frac{1}{9}\;\wfs(X)$. Then, for all $\alpha\in
\left[2\e,\;\frac{1}{4}\left(\wfs(X)-\e\right)\right]$ and all $\lambda\in
(0,\wfs(X))$, we have: $\forall k\in\N$, $H_k(X^\lambda)\isom \im
j_*$, where $j_*:H_k(\ripsa(\eee))\rightarrow
H_k(\rips{4\alpha}(\eee))$ is the homomorphism between homology groups
induced by the canonical inclusion
$j:\ripsa(\eee)\hookrightarrow\rips{4\alpha}(\eee)$.
\end{theorem}
\begin{proof}
From Lemma \ref{lem-cech-vs-rips} we deduce the following sequence of
inclusions:
\begin{equation}\label{eq-sandwich-rips}
\cech{\frac{\alpha}{2}}(\eee)\hookrightarrow\ripsa(\eee)\hookrightarrow
\cecha(\eee)\hookrightarrow
\cech{2\alpha}(\eee)\hookrightarrow\rips{4\alpha}(\eee)\hookrightarrow
\cech{4\alpha}(\eee)
\end{equation}
Since $\alpha \geq 2\e$, Theorem \ref{th-prop-cech} implies
that Eq. (\ref{eq-sandwich-rips}) induces a sequence of homomorphisms
between homology groups, such that 
$H_k(\cech{\frac{\alpha}{2}}(\eee))\rightarrow
H_k(\cech{4\alpha}(\eee))$ and $H_k(\cecha(\eee))\rightarrow
H_k(\cech{2\alpha}(\eee))$ have ranks equal to $\dim
H_k(X^\lambda)$. Therefore, by the Sandwich Lemma
\ref{sandwich-lemma}, $\rank j_*$ is also equal to $\dim
H_k(X^\lambda)$. It follows that $\im j_*\isom\dim H_k(X^\lambda)$,
since our ring of coefficients is a field.
\end{proof}

\noindent Similarly, Corollary \ref{cor-cech-vs-winf} provides the
following sequence of inclusions:
$$
\cech{\frac{\alpha}{4}}(\eee)\hookrightarrow\winfaw(\eee)\hookrightarrow
\cech{8\alpha}(\eee)\hookrightarrow\cech{9\alpha}(\eee)\hookrightarrow
\winfrw{36\alpha}(\eee)\hookrightarrow
\cech{288\alpha}(\eee),
$$
from which follows a result similar to Theorem \ref{th-prop-rips} on
the witness complex, by the same proof:
\begin{theorem}\label{th-prop-winf}
Let $X$ be a compact subset of $\R^d$, and let
  $\eee\subseteq\wit\subseteq\R^d$ be such that $\eee$ is
  finite. Assume that $\disth(X, \wit)\leq\delta$ and that
  $\disth(\wit, \eee)\leq\e$, with
  $\delta\leq\e<\min\left\{\frac{1}{8}\;\diamcc(X),\;
  \frac{1}{1153}\;\wfs(X)\right\}$. Then, for all
  $\alpha\in
  \left[4\e,\;\frac{1}{288}\left(\wfs(X)-\e\right)\right]$
  and all $\lambda\in (0,\wfs(X))$, we have: $\forall k\in\N$,
  $H_k(X^\lambda)\isom \im j_*$, where
  $j_*:H_k(\winfaw(\eee))\rightarrow H_k(\winfrw{36\alpha}(\eee))$ is
  the homomorphism between homology groups induced by the canonical
  inclusion $j:\winfaw(\eee)\hookrightarrow\winfrw{36\alpha}(\eee)$.
\end{theorem}
%
%
More generally, the above arguments show that the homology of
$X^\lambda$ can be recovered from the persistence barcode of any
filtration $\{F_\alpha\}_{\alpha\geq 0}$ that is intertwined with the
\v Cech filtration in the sense of Lemmas \ref{lem-cech-vs-rips} and
\ref{lem-cech-vs-winf}.
Note however that Theorems \ref{th-prop-rips} and \ref{th-prop-winf}
suggest a different behavior of the barcode in this case, since its
topological noise might scale up with $\alpha$ (specifically, it might
be up to linear in $\alpha$), whereas it is uniformly bounded by a
constant in the case of the \v Cech filtration. This difference of
behavior is easily explained by the way $\{F_\alpha\}_{\alpha\geq 0}$
is intertwined with the \v Cech filtration.
A trick to get a uniformly-bounded noise is to represent
the barcode of $\{F_\alpha\}_{\alpha\geq 0}$ on a logarithmic scale,
that is, with $\log_2\alpha$ instead of $\alpha$ in abcissa.

\subsection{Results on homotopy}

The results on homology obtained in Section \ref{sec-res-homology}
follow from simple algebraic arguments. Using a more geometric
approach, we can get similar results on homotopy. From now on, $x_0
\in X$ is a fixed point and all the homotopy groups $\pi_k(X) =
\pi_k(X,x_0)$ are assumed to be with base-point $x_0$.  Theorems
\ref{th-prop-rips} and \ref{th-prop-winf} can be extended to homotopy
in the following way:

\begin{theorem} \label{th-prop-rips-homotopy}
Under the same hypotheses as in Theorem \ref{th-prop-rips}, 
we have: $\forall k\in\N$, $\pi_k(X^\lambda)\isom \im
j_*$, where $j_*:\pi_k(\ripsa(\eee))\rightarrow
\pi_k(\rips{4\alpha}(\eee))$ is the homomorphism between homotopy groups
induced by the canonical inclusion
$j:\ripsa(\eee)\hookrightarrow\rips{4\alpha}(\eee)$.
\end{theorem}

\begin{theorem} \label{th-prop-winf-homotopy}
Under the same hypotheses as in Theorem \ref{th-prop-winf}, 
we have: $\forall k\in\N$,
  $\pi_k(X^\lambda)\isom \im j_*$, where
  $j_*:\pi_k(\winfaw(\eee))\rightarrow \pi_k(\winfrw{36\alpha}(\eee))$ is
  the homomorphism between homotopy groups induced by the canonical
  inclusion $j:\winfaw(\eee)\hookrightarrow\winfrw{36\alpha}(\eee)$.
\end{theorem}

The proofs of these two results being mostly identical, we focus
exclusively on the Rips complex. We will use the following
lemma, which is an immediate generalization of Proposition 4.1 of
\cite{cl-sctis-07}:

\begin{lemma} \label{lemma:k-loop}
Let $X$ be a compact set and $\eee$ a finite set in $\R^d$, such that
$\disth(X,\eee) < \e$ for some $\e < \frac{1}{4}\;\wfs(X)$. Let
$\alpha,\alpha'\in\left[\e, \wfs(X)-\e\right]$ be such that
$\alpha'-\alpha \geq 2\e$.  Given $k \in \N$, two {\em k-loops} $\sigma_1,
\sigma_2 : \mathbb{S}^k \rightarrow (L^\alpha,x_0)$ in $L^\alpha$ are
homotopic in $X^{\alpha'+\e}$ if and only if they are homotopic in
$L^{\alpha'}$.
\end{lemma}

\begin{proofarg}{of Theorem \ref{th-prop-rips-homotopy}}
As mentionned at the begining of the proof of Lemma
\ref{lem-prop-union-of-balls}, we can assume without loss of
generality that $2\e < \alpha < \frac{1}{4}(\wfs(X) - \e)$.  Consider
the following sequence of inclusions:
$$ \cech{\frac{\alpha}{2}}(L) \subset \rips{\alpha}(L) \subset
\cech{\alpha}(L) \subset \cech{2\alpha}(L) \subset \rips{4\alpha}(L)
\subset \cech{4\alpha}(L)$$ We use the homotopy equivalences
$h_\beta:L^\beta\rightarrow \cech{\beta}(L)$ provided by Lemma
\ref{lemma:parametric_nerve} for all values $\beta>0$, which commute
with inclusions at homotopy level. Note that, for any element $\sigma$
of $\pi_k(\cech{\beta}(L))$, there exists a $k$-loop in $L^\beta$ that
is mapped through $h_\beta$ to a $k$-loop representing the
homotopy class $\sigma$.  In the following, we denote by $\sigma_g$
such a $k$-loop.  Let $E, F$ and $G$ be the images of
$\pi_k(\cech{\frac{\alpha}{2}}(L))$ in $\pi_k(\cech{\alpha}(L))$,
$\pi_k(\cech{2\alpha}(L))$ and $\pi_k(\cech{4\alpha}(L))$ respectively,
through the homomorphisms induced by inclusion.  We thus have a sequence of
surjective homomorphisms:
$$\pi_k(\cech{\frac{\alpha}{2}}(L)) \rightarrow E \rightarrow F \rightarrow G$$
Note that, by Theorem \ref{th-prop-cech}, $F$ and $G$ are
isomorphic to $\pi_k(X^\lambda)$.
Let $\sigma \in F$ be a homotopy class. Since $F$ is the image of
$\pi_k(\cech{\frac{\alpha}{2}}(L))$, we can assume without loss of
generality that $\sigma_g \subset L^{\frac{\alpha}{2}}$. Assume that
the image of $\sigma$ in $G$ is equal to $0$. This means that
$\sigma_g$ is null-homotopic in $L^{4\alpha}$ and, since $L^{4\alpha}
\subset X^{4\alpha+\e}$, $\sigma_g$ is also null-homotopic in
$X^{4\alpha+\e}$. But $\sigma_g \subset L^{\frac{\alpha}{2}} \subset
X^{\frac{\alpha}{2}+\e}$, and $X^{2\alpha+\e}$ deformation retracts
onto $X^{\frac{\alpha}{2}+\e}$, by the Isotopy Lemma
\ref{lemma:isotopy}. As a consequence, $\sigma_g$ is null-homotopic in
$X^{\frac{\alpha}{2}+\e}$, which is contained in $L^{2\alpha}$ since
$\frac{\alpha}{2}+2\e < 2\alpha$. Hence, $\sigma_g$ is
null-homotopic in $L^{2\alpha}$, namely: $\sigma = 0$ in $F$. So, the
homomorphism $F \rightarrow G$ is injective, and thus it is an
isomorphism.
As a consequence, $F \rightarrow \pi_k(\rips{4\alpha}(L))$ is
injective, and it is now sufficient to prove that the image of
$\phi_*: \pi_k(\rips{\alpha}(L)) \rightarrow \pi_k(\cech{2\alpha}(L))$
induced by the inclusion is equal to $F$.

Obviously, $F$ is contained in the image of $\phi_*$. Now, let $\sigma
\in \pi_k(\rips{\alpha}(L))$ and let $\phi_*(\sigma)_g$ be a $k$-loop
in $L^{2\alpha}$ that is mapped through $h_{2\alpha}$ to a $k$-loop
representing the homotopy class $\phi_*(\sigma)$.  Since
$\phi_*(\sigma)$ is in the image of $\phi_*$, and since
$\rips{\alpha}(L) \subset \cech{\alpha}(L)$, we can assume that
$\phi_*(\sigma)_g$ is contained in $L^\alpha$. Let $\tilde \sigma_g$
be the image of $\phi_*(\sigma)_g$ through a deformation retraction of
$X^{2\alpha+\e}$ onto $X^{\alpha_0}$, where $0< \alpha_0 <
\frac{\alpha}{2}$ is such that $\frac{\alpha}{2}-\alpha_0 >
\e$. Obviously, $\tilde \sigma_g$ and $\phi_*(\sigma)_g$ are homotopic
in $X^{2\alpha+\e}$. It follows then from Lemma \ref{lemma:k-loop}
that $\tilde \sigma_g$ and $\phi_*(\sigma)_g$ are homotopic in
$L^{2\alpha}$. And since $\tilde \sigma_g$ is contained in
$X^{\alpha_0} \subset L^{\frac{\alpha}{2}}$, the equivalence class of
$h_{\frac{\alpha}{2}}(\tilde\sigma_g)$ in
$\pi_k(\cech{\frac{\alpha}{2}}(L))$ is mapped to
$\phi_*(\sigma)\in\pi_k(\cech{2\alpha}(L))$ through the homomorphism
induced by
$\cech{\frac{\alpha}{2}}(\eee)\hookrightarrow\cech{2\alpha}(\eee)$,
which commutes with the homotopy equivalences.  As a result,
$\phi_*(\sigma)$ belongs to $F$, which is therefore equal to
$\im\phi_*$.
\end{proofarg}

%
\section{The case of smooth submanifolds of $\R^d$}
\label{sec-smooth-case}

In this section, we consider the case of submanifolds $\sss$ of $\R^d$
that have positive {\em reach}. Recall that the reach of $\sss$, or
$\rch(\sss)$ for short, is the minimum distance between the points of
$\sss$ and the points of its medial axis \cite{ab-srvf-99}. A point
cloud $\eee\subset\sss$ is an {\em $\e$-sample} of $\sss$ if every
point of $\sss$ lies within distance $\e$ of $\eee$. In addition,
$\eee$ is {\em $\e$-sparse} if its points lie at least $\e$ away from
one another.

Our main result is a first attempt at quantifying a conjecture of
Carlsson and de Silva \cite{cds-teuwc-04}, according to which the
witness complex filtration should have {\em cleaner} persistence
barcodes than the \v Cech and Rips filtrations, at least on smooth
submanifolds of $\R^d$. By {\em cleaner} is meant that the amplitude
of the topological noise in the barcodes should be smaller, and also
that the long intervals should appear earlier. We prove this latter
statement correct, at least to some extent:
\begin{theorem}\label{th-prop2-winf}
There exist a constant $\varrho>0$ and a continuous, non-decreasing
map $\wpb:[0,\varrho)\rightarrow[0,\frac{1}{2})$, such that, for any
    submanifold $\sss$ of $\R^d$, for all $\e,\delta$ satisfying
    $0<\delta\leq\e<\varrho\;\rch(\sss)$, for any $\delta$-sample
    $\wit$ of $\sss$ and any $\e$-sparse $\e$-sample $\eee$ of $\wit$,
    $\winfaw(\eee)$ contains a subcomplex $\cal D$ homeomorphic to
    $\sss$ and such that the canonical inclusion ${\cal
      D}\hookrightarrow\winfaw(\eee)$ induces an injective homomorhism
    between homology groups, provided that $\alpha$ satisfies:
    $\frac{8}{3}(\delta+\wpb(\frac{\e}{\rch(\sss)})^2\e) \leq\alpha <
    \frac{1}{2}\;\rch(\sss)- (3+\frac{\sqrt{2}}{2})(\e+\delta)$.
\end{theorem}
This theorem guarantees that, for values of $\alpha$ ranging from
$O(\delta+\wpb(\frac{\e}{\rch(\sss)})^2\e)$ to $\Omega(\rch(\sss))$,
the topology of $\sss$ is captured by a subcomplex $\cal D$ that
injects itself suitably in $\winfaw(\eee)$. As a result, long
intervals showing the homology of $\sss$ appear around
$\alpha=O(\delta+\wpb(\frac{\e}{\rch(\sss)})^2\e)$ in the persistence
barcode of the witness complex filtration.
This can be much sooner than the time $\alpha=2\e$ prescribed by
Theorem \ref{th-prop-winf}, since $\wpb(\frac{\e}{\rch(\sss)})$ can be
arbitrarily small. Specifically, the denser the landmark set $\eee$,
the smaller the ratio $\frac{\e}{\rch(\sss)}$, and therefore the
smaller $\delta+\wpb(\frac{\e}{\rch(\sss)})^2\e$ compared to $2\e$. We
have reasons to believe that this upper bound on the appearance time
of long bars is tight. In particular, the bound
cannot depend solely on $\delta$, since otherwise, in the limit case
where $\delta=0$, we would get that the homology groups of $\sss$ can
be injected into the ones of the standard witness complex
$\winfw(\eee)$, which is known to be false \cite{go-ruwc-07,
  o-ntrdwchd-07}. The same argument implies that the amplitude of the
topological noise in the barcode cannot depend solely on $\delta$
either. However, whether the upper bound $O(\e)$ on the
amplitude of the noise can be improved or not is still an open question.

Our proof of Theorem \ref{th-prop2-winf} generalizes and argument used
in \cite{ggow-gdtp-08} for the planar case, which stresses the close
relationship that exists between the $\alpha$-witness complex and the
so-called {\em weighted restricted Delaunay triangulation}
$\wdels(\eee)$. Given a submanifold $\sss$ of $\R^d$, a finite
landmark set $\eee\subset\R^d$, and an assignment of non-negative
weights to the landmarks, specified through a map
$\weight:\eee\rightarrow [0,\infty)$, $\wdels(\eee)$ is the nerve of
  the restriction to $\sss$ of the {\em power diagram}\footnote{More
    on power diagrams and on restricted triangulations can be found in
    \cite{a-vdsfg-91} and \cite{es-tts-97} respectively.} of the
  weighted set $\eee$. Under the hypotheses of the theorem, we show
  that $\winfaw(\eee)$ contains $\wdels(\eee)$, which, by a result of
  Cheng {\em et al.}  \cite{cdr-mrps-2005}
(see Theorem \ref{th-wdels-homeo-S} below),
is homeomorphic to $\sss$. The main point of the proof is then to
show that $\wdels(\eee)$ injects itself {\em nicely} into
$\winfaw(\eee)$.


The rest of the section is devoted to the proof of
Theorem \ref{th-prop2-winf}.  After introducing the weighted
restricted Delaunay triangulation in Section \ref{sec-restr-del} and
stressing its relationship with the $\alpha$-witness complex in
Section \ref{sec-dels-winf}, we detail the proof of Theorem
\ref{th-prop2-winf} in Section \ref{sec-proof-th-prop2-winf}.

\subsection {The weighted restricted Delaunay triangulation}
\label{sec-restr-del}
Given a finite point set $\eee\subset\R^{d}$, an {\em assignment of
  weights over $\eee$} is a non-negative real-valued function
$\weight: \eee\rightarrow [0,\infty)$. The quantity $\max_{u\in\eee,
    v\in\eee\setminus\{u\}} \frac{\weight(u)}{\|u-v\|}$ is called the
  {\em relative amplitude} of $\weight$. Given $p\in\R^{d}$, the {\em
    weighted distance} from $p$ to some weighted point $v\in\eee$ is
  $\|p-v\|^2-\weight(v)^2$. This is actually not a metric, since it is
  not symmetric.
Given a finite point set $\eee$ and an assignment of weights $\weight$
over $\eee$, we denote by $\wvor(\eee)$ the power diagram of the
weighted set $\eee$, and by $\wdel(\eee)$ its nerve, also known as 
the weighted Delaunay triangulation. If the
relative amplitude of $\weight$ is at most $\frac{1}{2}$, then the
points of $\eee$ have non-empty cells in $\wvor(\eee)$, and in fact
each point of $\eee$ belongs to its own cell \cite{cdeft-se-00}. For
any simplex $\simplex$ of $\wdel(\eee)$, $\wdual{\simplex}$ denotes
the face of $\wvor(\eee)$ dual to $\simplex$.
%

Given a subset $X$ of $\R^{d}$, we call $\wvorsarg{X}(\eee)$ the
restriction of $\wvor(\eee)$ to $X$, and we denote by
$\wdelsarg{X}(\eee)$ its nerve, also known as the weighted Delaunay
triangulation of $\eee$ restricted to $X$. Observe that
$\wdelsarg{X}(\eee)$ is a subcomplex of $\wdel(\eee)$. In the special
case where all the weights are equal, $\wvor(\eee)$ and $\wdel(\eee)$
coincide with their standard Euclidean versions, $\vor(\eee)$ and
$\del(\eee)$. Similarly, $\wdual{\simplex}$ becomes $\dual{\simplex}$,
and $\wvorsarg{X}(\eee)$ and $\wdelsarg{X}(\eee)$ become respectively
$\vorsarg{X}(\eee)$ and $\delsarg{X}(\eee)$.
\begin{theorem}[Lemmas 13, 14, 18 of \cite{cdr-mrps-2005}, 
see also Theorem 2.5 of \cite{bgo-mradwc-07}]
\label{th-wdels-homeo-S}
There exist\footnote{Note that $\varrho$ and $\wpb$ are the same as in
  Theorem \ref{th-prop2-winf}. In fact, these quantities come from
  Theorem \ref{th-wdels-homeo-S}.} a constant $\varrho>0$
and a non-decreasing continuous map
$\wpb:[0,\varrho)\rightarrow[0,\frac{1}{2})$, such that, for any
    manifold $\sss$ and any $\e$-sparse $2\e$-sample $\eee$ of $\sss$,
    with $\e< \varrho\;\rch(\sss)$, there is an assignment of weights
    $\weight$ of relative amplitude at most
    $\wpb\left(\frac{\e}{\rch(\sss)}\right)$ such that $\wdels(\eee)$
    is homeomorphic to~$\sss$.
\end{theorem}
This theorem guarantees that the topology of $\sss$ is captured by
$\wdels(\eee)$ provided that the landmarks are sufficiently densely
sampled on $\sss$, and that they are assigned suitable
weights. Observe that the denser the landmark set, the smaller the
weights are required to be, as specified by the map $\wpb$. In the
particular case where $\sss$ is a curve or a surface, $\wpb$ can be
taken to be the constant zero map, since $\dels(\eee)$ is homeomorphic
to $\sss$ \cite{ab-srvf-99, abe-cbscc-98}.  On higher-dimensional
manifolds though, positive weights are required, since $\dels(\eee)$
may fail to capture the topological invariants of $X$
\cite{o-ntrdwchd-07}.

The proof of the theorem given in \cite{cdr-mrps-2005} shows that
$\wvors(\eee)$ satisfies the so-called {\em closed ball property},
which states that every face of the weighted Voronoi diagram
$\wvor(\eee)$ intersects the manifold $\sss$ along a topological ball
of proper dimension, if at all. Under this condition, there exists a
homeomorphism $h_0$ between the nerve $\wdels(\eee)$ and $\sss$, as
proved by Edelsbrunner and Shah \cite{es-tts-97}. Furthermore, $h_0$
sends every simplex of $\wdels(\eee)$ to a subset of the union of the
restricted Voronoi cells of its vertices, that is: $\forall
\simplex\in\wdels(\eee)$, $h_0(\simplex)\subseteq\bigcup_{v\mbox{
    \scriptsize vertex of }\simplex} \wdual{v}\cap\sss$. This fact
will be instrumental in the proof of Theorem \ref{th-prop2-winf}.

\subsection{Relationship between $\wdels(\eee)$ and $\winfaw(\eee)$}
\label{sec-dels-winf}
As mentioned in introduction, the use of the witness complex
filtration for topological data analysis is motivated by its close
relationship with the weighted restricted Delaunay triangulation:
\begin{lemma}\label{lem-wdels-in-winf}
Let $\sss$ be a compact subset of $\R^d$, $\wit\subseteq\sss$ a
$\delta$-sample of $\sss$, and $\eee\subseteq\wit$ an $\e$-sparse
$\e$-sample of $\wit$. Then, for all assignment of weights $\weight$
of relative amplitude $\wpb\leq\frac{1}{2}$, $\wdels(\eee)$ is
included in $\winfaw(\eee)$ whenever $\alpha\geq
\frac{2}{1-\wpb^2}\left(\delta+\wpb^2\e\right)$.
\end{lemma}
This result implies in particular that $\dels(\eee)$ is included in
$\winfaw(\eee)$ whenever $\alpha\geq 2\delta$, since $\dels(\eee)$ is
nothing but $\wdels(\eee)$ for an assignment of weights of relative
amplitude zero.

\begin{proof}
Let $\simplex$ be a simplex of $\wdels(\eee)$. If $\simplex$ is a
vertex, then it clearly belongs to $\winfaw(\eee)$ for all $\alpha\geq
0$, since $\eee\subseteq\wit$. Assume now that $\simplex$ has positive
dimension, and consider a point $c\in\wdual\simplex\cap\sss$. For any
vertex $v$ of $\simplex$ and any point $p$ of $\eee$ (possibly equal
to $v$), we have: $\|v-c\|^2-\weight(v)^2\leq\|p-c\|^2-\weight(p)^2$,
which yields: $\|v-c\|^2\leq \|p-c\|^2+
\weight(v)^2-\weight(p)^2$. Now, $\weight(p)^2$ is non-negative, while
$\weight(v)^2$ is at most $\wpb^2\|v-p\|^2$, which gives:
$\|v-c\|^2\leq \|p-c\|^2+\wpb^2\|v-p\|^2$. Replacing $\|v-p\|$ by
$\|v-c\|+\|p-c\|$, we get a
semi-algebraic expression of degree 2 in $\|v-c\|$, namely:
$
(1-\wpb^2)\|v-c\|^2
-2\wpb^2\|p-c\|\|v-c\|-(1+\wpb^2)\|p-c\|^2\leq 0.
$
It follows that $\|v-c\|\leq\frac{1+\wpb^2}{1-\wpb^2}\;\|p-c\|$. Let
now $w$ be a point of $\wit$ closest to $c$ in the Euclidean
metric. Using the triangle inequality and the fact that
$\|w-c\|\leq\delta$, we get: $\|v-w\|\leq\|v-c\|+\|w-c\|\leq
\frac{1+\wpb^2}{1-\wpb^2}\;\|p-c\| + \delta$. This holds for any point
$p\in\eee$, and in particular for the nearest neighbor $p_w$ of $w$ in
$\eee$. Therefore, we have $\|v-w\|\leq
\frac{1+\wpb^2}{1-\wpb^2}\;\|p_w-c\|+\delta$, which is at most
$\frac{1+\wpb^2}{1-\wpb^2}\;(\|p_w-w\|+\delta)+\delta\leq\|p_w-w\|+
\frac{2}{1-\wpb^2}\left(\delta+\wpb^2\e\right)$ because
$\|w-c\|\leq\delta$ and $\|w-p_w\|\leq\e$. Since this inequality holds
for any vertex $v$ of $\simplex$, and since the Euclidean distances
from $w$ to all the landmarks are at least $\|p_w-w\|$, $w$ is an
$\alpha$-witness of $\simplex$ and of all its faces as soon as
$\alpha\geq \frac{2}{1-\wpb^2}\left(\delta+\wpb^2\e\right)$. Since
this holds for every simplex $\simplex$ of $\wdels(\eee)$, the lemma
follows.
\end{proof}

\subsection{Proof of Theorem \ref{th-prop2-winf}}
\label{sec-proof-th-prop2-winf}
The proof is mostly algebraic, but it relies on two technical
results. The first one is Dugundji's extension theorem
\cite{d-ett-51}, which states that, given an abstract simplex
$\simplex$ and a continuous map $f:\partial\simplex\rightarrow\R^d$,
$f$ can be extended to a continuous map $f:\simplex\rightarrow\R^d$
such that $f(\simplex)$ is included in the Euclidean convex hull of
$f(\partial\simplex)$, noted $\conv(f(\partial\simplex))$. This
convexity property of $f$ is used in the proof of the second technical
result, stated as Lemma \ref{lem-no-intersection-with-ma} and proved
at the end of the section.

\begin{proofarg}{of Theorem \ref{th-prop2-winf}}
Since $\delta\leq\e$, $\eee$ is an $\e$-sparse $2\e$-sample of $\sss$,
with $\e<\varrho\;\rch(\sss)$. Therefore, by Theorem
\ref{th-wdels-homeo-S}, there exists an assignment of weights
$\weight$ over $\eee$, of relative amplitude at most
$\wpb\left(\frac{\e}{\rch(\sss)}\right)$, such that $\wdels(\eee)$ is
homeomorphic to $\sss$. Taking ${\cal D}=\wdels(\eee)$, we then have:
$\forall k\in\N$, $H_k(X)\isom H_k({\cal D})$. Moreover, by Lemma
\ref{lem-wdels-in-winf}, we know that ${\cal D}=\wdels(\eee)$ is
included in $\winfaw(\eee)$, since $\alpha\geq
\frac{8}{3}\left(\wpb\left(\frac{\e}{\rch(\sss)}\right)^2\e+\delta\right)\geq
\frac{2}{1-\wpb\left(\frac{\e}{\rch(\sss)}\right)^2}
\left(\wpb\left(\frac{\e}{\rch(\sss)}\right)^2\e+\delta\right)$. There
 remains to show that the inclusion map
$j:\wdels(\eee)\hookrightarrow\winfaw(\eee)$ induces injective
homomorphisms $j_*$ between the homology groups of $\wdels(\eee)$ and
$\winfaw(\eee)$, which will conclude the proof of the theorem.

Our approach to showing the injectivity of $j_*$ consists in building
a continuous map\footnote{Note that this map does not need to be
  simplicial, since we are using singular homology.}
$h:\winfaw(\eee)\rightarrow\wdels(\eee)$ such that $h\circ j$ is
homotopic to the identity in $\wdels(\eee)$. This implies that
$h_*\circ j_*:H_k(\wdels(\eee))\rightarrow H_k(\wdels(\eee))$ is an
isomorphism (in fact, it is the identity map), and thus that $j_*$ is
injective.

We begin our construction with the homeomorphism
$h_0:\wdels(\eee)\rightarrow\sss$ provided by the theorem of
Edelsbrunner and Shah \cite{es-tts-97}. Taking $h_0$ as a map
$\wdels(\eee)\rightarrow\R^d$, we extend
it to a continuous map $\tilde h_0:\winfaw(\eee)\rightarrow\R^d$ by
the following iterative process: while there exists a simplex
$\simplex\in\winfaw(\eee)$ such that $\tilde h_0$ is defined over the
boundary of $\simplex$ but not over its interior, apply Dugundji's
extension theorem, which extends $\tilde h_0$ to the
entire simplex $\simplex$.
\begin{lemma}\label{lem-iterative-proc}
The above iterative process extends $h_0$ to a map $\tilde
h_0:\winfaw(\eee)\rightarrow\R^d$.
\end{lemma}
\begin{proof}
We only need to prove that the process visits every simplex of
$\winfaw(\eee)$. Assume for a contradiction that the process
terminates while there still remain some unvisited simplices of
$\winfaw(\eee)$. Consider one such simplex $\simplex$ of minimal
dimension. Either $\simplex$ is a vertex, or there is at least one
proper face of $\simplex$ that has not yet been visited -- since
otherwise the process could visit $\simplex$. In the former case,
$\simplex$ is a point of $\eee$, and as such it is a
vertex\footnote{Indeed, every point $p\in\eee$ lies on $\sss$ and
  belongs to its own cell, since $\weight$ has relative amplitude less
  than $\frac{1}{2}$. Therefore, $\wdual{p}\cap\sss\neq\emptyset$,
  which means that $p$ is a vertex of $\wdels(\eee)$.} of
$\wdels(\eee)$, which means that $h_0$ is already defined over
$\simplex$ (contradiction). In the latter case, we get a contradiction
with the fact that $\simplex$ is of minimal dimension.
\end{proof}


Now that we have built a map $\tilde
h_0:\winfaw(\eee)\rightarrow\R^d$, our next step is to turn it into a
map $\winfaw(\eee)\rightarrow\sss$. To do so, we compose it with the
projection $\projS$ that maps every point of $\R^d$ to its nearest
neighbor on $\sss$, if the latter is unique. This projection is
known to be well-defined and continuous over $\R^d\setminus\ma$, where
$\ma$ denotes the medial axis of $\sss$ \cite{f-gm-59}.
\begin{lemma}\label{lem-no-intersection-with-ma}
Let $\sss, \wit, \eee, \delta, \e$ satisfy the hypotheses of Theorem
\ref{th-prop2-winf}. Then, $\tilde
h_0(\winfaw(\eee))\cap\ma=\emptyset$ as long as
$\alpha<\frac{1}{2}\;\rch(\sss)-
\left(3+\frac{\sqrt{2}}{2}\right)(\e+\delta)$.
\end{lemma}
Since by Lemma \ref{lem-no-intersection-with-ma} we have $\tilde
h_0(\winfaw(\eee))\cap\ma=\emptyset$, the map $\projS\circ\tilde
h_0:\winfaw(\eee)\rightarrow\sss$ is well-defined and continuous. Our
final step is to compose it with $h_0^{-1}$, to get a continuous map
$h=h_0^{-1}\circ\projS\circ\tilde h_0:
\winfaw(\eee)\rightarrow\wdels(\eee)$. The restriction of $h$ to
$\wdels(\eee)$ is simply $h_0^{-1}\circ\projS\circ h_0$, which
coincides with $h_0^{-1}\circ h_0=\id$ since
$h_0(\wdels(\eee))=\sss$. It follows that $h\circ j$ is homotopic to
the identity in $\wdels(\eee)$ (in fact, it is the identity), and
therefore that the induced map $h_*\circ j_*$ is the identity. This
implies that $j_*:H_k(\wdels(\eee))\rightarrow H_k(\winfaw(\eee))$ is
injective, which concludes the proof of Theorem \ref{th-prop2-winf}.
\end{proofarg}

We end the section by providing the proof of Lemma
\ref{lem-no-intersection-with-ma}:

\begin{proofarg}{of Lemma \ref{lem-no-intersection-with-ma}}
First, we claim that the image through $\tilde h_0$ of any simplex of
$\winfaw(\eee)$ is included in the Euclidean convex hull of the
restricted Voronoi cells of its simplices, that is: $\forall
\simplex\in\winfaw(\eee)$, $\tilde h_0(\simplex)\subseteq \conv
\left(\bigcup_{v\mbox{ \scriptsize vertex of }\simplex}
\wdual{v}\cap\sss\right)$. This is clearly true if $\simplex$ belongs
to $\wdels(\eee)$, since in this case we have $\tilde
h_0(\simplex)=h_0(\simplex)\subseteq \bigcup_{v\mbox{ \scriptsize
    vertex of }\simplex} \wdual{v}\cap\sss$, as mentioned after
Theorem \ref{th-wdels-homeo-S}. Now, if the property holds for all the
proper faces of a simplex $\simplex\in\winfaw(\eee)$, then by
induction it also holds for the simplex itself. Indeed, for each
proper face $\tau\subset\simplex$, we have $\tilde
h_0(\tau)\subseteq \conv\left(\bigcup_{v\mbox{ \scriptsize vertex of
  }\tau} \wdual{v}\cap\sss\right)\subseteq \conv\left(\bigcup_{v\mbox{
    \scriptsize vertex of }\simplex}
\wdual{v}\cap\sss\right)$. Therefore, $\conv\left(\bigcup_{v\mbox{
    \scriptsize vertex of }\simplex} \wdual{v}\cap\sss\right)$
contains $\conv \left(\tilde h_0(\partial\simplex)\right)$, which, by
Dugundji's extension theorem, contains $\tilde
h_0(\simplex)$. Therefore, the property holds for every simplex of
$\winfaw(\eee)$.

We can now prove that the image through $\tilde h_0$ of any arbitrary
simplex $\simplex$ of $\winfaw(\eee)$ does not intersect the medial
axis of $\sss$. This is clearly true if $\simplex$ is a simplex of
$\wdels(\eee)$, since in this case $\tilde
h_0(\simplex)=h_0(\simplex)$ is included in $\sss$. Assume now that
$\simplex\notin\wdels(\eee)$. In particular, $\simplex$ is not a
vertex. Let $v$ be an arbirtary vertex of $\simplex$. Consider any
other vertex $u$ of $\simplex$. Edge $[u,v]$ is $\alpha$-witnessed by
some point $w_{uv}\in\wit$. We then have
$\|v-u\|\leq\|v-w_{uv}\|+\|w_{uv}-u\|\leq 2\dist_2(w_{uv})+2\alpha$,
where $\dist_2(w_{uv})$ stands for the Euclidean distance from
$w_{uv}$ to its second nearest landmark. According to Lemma 3.4 of
\cite{bgo-mradwc-07}, we have $\dist_2(w)\leq 3(\e+\delta)$, since
$\eee$ is an $(\e+\delta)$-sample of $\sss$. Thus, all the vertices of
$\simplex$ are included in the Euclidean ball $B(v,
2\alpha+6(\e+\delta))$.  Moreover, for any vertex $u$ of $\simplex$
and any point $p\in\wdual{u}\cap\sss$, we have
$\|p-u'\|\leq\e+\delta$, where $u'$ is a landmark closest to $p$ in
the Euclidean metric. Combined with the fact that
$\|p-u\|^2-\weight(u)^2\leq \|p-u'\|^2-\weight(u')^2$, we get:
$\|p-u\|^2\leq \|p-u'\|^2 +\weight(u)^2\leq 2(\e+\delta)^2$, since by
Lemma 3.3 of \cite{bgo-mradwc-07} we have $\weight(u)\leq
2\;\wpb\left(\frac{\e}{\rch(\sss)}\right)(\e+\delta)\leq\e+\delta$. Hence,
$\wdual{u}\cap\sss$ is included in $B(u,\sqrt{2}(\e+\delta))\subset
B(v, 2\alpha +(6+\sqrt{2})(\e+\delta))$. Since this is true for every
vertex $u$ of $\simplex$, we get: $\tilde h_0(\simplex)\subseteq \conv
\left(\bigcup_{u\mbox{ \scriptsize vertex of }\simplex}
\wdual{u}\cap\sss\right) \subseteq B(v,
2\alpha+(6+\sqrt{2})(\e+\delta))$. Now, $v$ belongs to
$\eee\subseteq\wit\subseteq\sss$, and by assumption we have
$2\alpha+(6+\sqrt{2})(\e+\delta) <\rch(\sss)$, therefore $\tilde
h_0(\simplex)$ does not intersect the medial axis of $\sss$.
\end{proofarg}

%
\section{Application to reconstruction}
\label{sec-recons}

Taking advantage of the structural results of Section
\ref{sec-struct-res-compact}, we devise a very simple yet
provably-good algorithm for constructing nested pairs of complexes
that can capture the homology of a large class of compact subsets of
$\R^d$. This algorithm is a variant of the greedy refinement technique
of \cite{go-ruwc-07}, which builds a set $\eee$ of landmarks
iteratively, and in the meantime maintains a suitable data
structure. In our case, the data structure is composed of a nested
pair of simplicial complexes, which can be either
$\ripsa(\eee)\hookrightarrow\ripsap(\eee)$ or
$\winfaw(\eee)\hookrightarrow\winfapw(\eee)$, for specific values
$\alpha<\alpha'$. Both variants of the algorithm can be used in
arbitrary metric spaces, with similar theoretical guarantees, although
the variant using witness complexes is likely to be more effective in
practice. In the sequel we focus on the variant using Rips complexes
because its analysis is somewhat simpler.

\subsection{The algorithm}
\label{sec-algo}

The input is a finite point set $\wit$ drawn from an arbitrary metric
space, together with the pairwise distances $l(w,w')$ between the
points of $\wit$. In the sequel, $\wit$ is identified as the set of
witnesses.  

Initially, $\eee=\emptyset$ and $\e=+\infty$.  At each iteration, the
point of $\wit$ lying furthest away\footnote{At the first iteration,
  since $\eee$ is empty, an arbitrary point of $\wit$ is chosen.} from
$\eee$ in the metric $l$ is inserted in $\eee$, and $\e$ is set to
$\max_{w\in\wit}\min_{v\in\eee} l(w,v)$. Then, $\ripsqe(\eee)$ and
$\ripsse(\eee)$ are updated, and the persistent homology of
$\ripsqe(\eee)\hookrightarrow\ripsse(\eee)$ is computed using the
persistence algorithm \cite{zc-cph-05}.
The algorithm terminates when $\eee=\wit$. The output is the diagram
showing the evolution of the persistent Betti numbers versus $\e$,
which have been maintained throughout the process. As we will see in
Section \ref{sec-guarantees} below, with the help of this diagram the
user can determine a relevant scale at which to process the data: it
is then easy to generate the corresponding subset $\eee$ of landmarks
(the points of $\wit$ have been sorted according to their order of
insertion in $\eee$ during the process), and to rebuild
$\ripsqe(\eee)$ and $\ripsse(\eee)$.  The pseudo-code of the algorithm
is given in Figure \ref{fig-algo}.

\begin{figure}[htb]
\centering
\fbox{
\begin{minipage}{0.8\linewidth}
{\bf Input:} $\wit$ finite, together with distances $l(w,w')$ for all
$w,w'\in\wit$.

{\bf Init:} Let $\eee := \emptyset$, $\e:=+\infty$;

While $\eee\subsetneq\wit$ do

\sind\begin{minipage}{\linewidth}

Let $p:={\rm argmax}_{w\in\wit} \min_{v\in\eee} l(w,v)$;
~{\em // $p$ chosen arbitrarily in $\wit$ if $\eee=\emptyset$}

$\eee := \eee\cup\{p\}$;

$\e:=\max_{w\in\wit} \min_{v\in\eee} l(w,v)$;

Update $\ripsqe(\eee)$ and $\ripsse(\eee)$;

Compute persistent homology of
$\ripsqe(\eee)\hookrightarrow\ripsse(\eee)$;

\end{minipage}\\

End\_while

{\bf Output:} diagram showing the evolution of persistent Betti numbers
versus $\e$.
\end{minipage}
}
\caption{Pseudo-code of the algorithm.}
\label{fig-algo}
\end{figure}

\subsection{Guarantees on the output}
\label{sec-guarantees}

For any $i>0$, let $\eeei$ and $\ei$ denote respectively $\eee$ and
$\e$ at the end of the $i$th iteration of the main loop of the
algorithm.  Since $\eeei$ keeps growing with $i$, $\ei$ is a
decreasing function of $i$. In addition, $\eeei$ is an $\ei$-sample of
$\wit$, by definition of $\ei$.
Hence, if $\wit$ is a $\delta$-sample of some compact set
$\sss\subset\R^d$, then $\eeei$ is a ($\delta+\ei$)-sample of
$\sss$. This quantity is less than $2\ei$ whenever $\ei>\delta$.
Therefore, Theorem \ref{th-prop-rips} provides us with the following
theoretical guarantee:
\begin{theorem}\label{th-guarantees}
Assume that the input point set $\wit$ is a $\delta$-sample of some
compact set $X\subset\R^d$, with $\delta<\frac{1}{18}\wfs(X)$.
Then, at each iteration $i$ such that
$\delta<\ei<\frac{1}{18}\wfs(X)$, the persistent homology groups
of $\ripsqi(\eeei)\hookrightarrow\ripssi(\eeei)$ are isomorphic to the
homology groups of $X^\lambda$, for all $\lambda\in
(0,\wfs(X))$.
\end{theorem}
This theorem ensures that, when the input point cloud $\wit$ is
sufficiently densely sampled from a compact set $X$, there exists a
range of values of $\e(i)$ such that the persistent Betti numbers of
$\ripsqi(\eeei)\hookrightarrow\ripssi(\eeei)$ coincide with the ones
of sufficiently small offsets $X^\lambda$. This means that a plateau
appears in the diagram of persistent Betti numbers, showing the Betti
numbers of $X^\lambda$. In view of Theorem \ref{th-guarantees}, the
width of the plateau is at least $\frac{1}{18}\wfs(X)-\delta$. The
theorem also tells where the plateau is located in the diagram, but in
practice this does not help since neither $\delta$ nor $\wfs(X)$ are
known. However, when $\delta$ is small enough compared to $\wfs(X)$,
the plateau is large enough to be detected (and thus the homology of
small offsets of $X$ inferred) by the user or a software agent. In
cases where $\wit$ samples several compact sets with different weak
feature sizes, Theorem \ref{th-guarantees} ensures that several
plateaus appear in the diagram, showing plausible reconstructions at
various scales -- see Figure \ref{fig-spiral} (right). These
guarantees are similar to the ones provided with the low-dimensional
version of the algorithm \cite{go-ruwc-07}.

Once one or more plateaus have been detected, the user can choose a
relevant scale at which to process the data: as mentioned in Section
\ref{sec-algo} above, it is then easy to generate the corresponding
set of landmarks and to rebuild $\ripsqe(\eee)$ and
$\ripsse(\eee)$. Differently from the algorithm of \cite{go-ruwc-07},
the outcome is not a single embedded simplicial complex, but a nested
pair of abstract complexes whose images in $\R^d$ lie at Hausdorff
distance\footnote{Indeed, every simplex of $\ripsse(\eee)$ has all its
  vertices in $\sss^{\e+\delta}\subseteq \sss^{2\e}$, and the lengths
  of its edges are at most~$16\e$.} $O(\e)$ of $X$, such that the
persistent homology of the nested pair coincides with the homology of
$X^\lambda$.

\subsection{Update of $\ripsqe(\eee)$ and $\ripsse(\eee)$}
\label{sec-update}

We will now describe how to maintain $\ripsqe(\eee)$ and
$\ripsse(\eee)$. In fact, we will settle for describing how to rebuild
$\ripsse(\eee)$ completely at each iteration, which is sufficient for
achieving our complexity bounds. In practice, it would be much
preferable to use more local rules to update the simplicial complexes,
in order to avoid a complete rebuilding at each iteration.

Consider the one-skeleton graph $G$ of $\ripsse(\eee)$. The vertices
of $G$ are the points of $\eee$, and its edges are the sets
$\{p,q\}\subseteq\eee$ such that $\|p-q\|\leq 16\e$. Now, by
definition, a simplex that is not a vertex belongs to $\ripsse(\eee)$
if and only if all its edges are in $\ripsse(\eee)$. Therefore, the
simplices of $\ripsse(\eee)$ are precisely the cliques of $G$. The
simplicial complex can then be built as follows:
\begin{pliste}
\item[1.] build graph $G$, 
\item[2.] find all maximal cliques in $G$,
\item[3.] report the maximal cliques and all their subcliques.
\end{pliste}
Step 1. is performed within 
$O(|\eee|^2)$ time by checking the distances between all pairs of
landmarks. Here, $|G|$ denotes the size of $G$ and $|\eee|$ the size
of $\eee$. To perform Step 2., we use the output-sensitive algorithm
of \cite{tias-77}, which finds all the maximal cliques of $G$ in
$O(k\;|\eee|^3)$ time, where $k$ is the size of the answer. 
Finally, reporting all the subcliques of the maximal cliques is done
in
time linear in the total number of cliques, which is also the size of
$\ripsse(\eee)$. Therefore,
\begin{corollary}\label{cor-update-ripsa}
At each iteration of the algorithm, $\ripsqe(\eee)$ and
$\ripsse(\eee)$ are rebuilt within 
$O(|\ripsse(\eee)|\;|\eee|^3)$ time, where $|\ripsse(\eee)|$ is the size of
$\ripsse(\eee)$ and $|\eee|$ the size of~$\eee$.
\end{corollary}

\subsection{Running time of the algorithm}
\label{sec-complexity}

Let $|\wit|, |\eee|, |\ripsse(\eee)|$ denote the sizes of $\wit, \eee,
\ripsse(\eee)$ respectively.  At each iteration, point $p$ and
parameter $\e$ are computed naively by iterating over the witnesses,
and for each witness, by reviewing its distances to all the
landmarks. This procedure takes 
$O(|\wit||\eee|)$ time. According to Corollary \ref{cor-update-ripsa},
$\ripsqe(\eee)$ and $\ripsse(\eee)$ are updated (in fact, rebuilt) in
$O(|\ripsse(\eee)||\eee|^3)$ time. Finally, the persistence algorithm
runs in 
$O(|\ripsse(\eee)|^3)$ time \cite{elz-tps-02, zc-cph-05}. Hence,
%
\begin{lemma}\label{lem-one-iteration}
The running time
of one iteration of the algorithm 
is
$O(|\wit||\eee|+|\ripsse(\eee)||\eee|^3+|\ripsse(\eee)|^3)$.
\end{lemma}
There remains to find a reasonable bound on the size of
$\ripsse(\eee)$, which can be done in Euclidean space $\R^d$,
especially when the landmarks lie on a smooth submanifold:
\begin{lemma}\label{lem-size}
Let $\eee$ be a finite $\e$-sparse point set in $\R^d$. Then,
$\ripsse(\eee)$ has at most $2^{33^d}|\eee|$ simplices. If in addition
the points of $\eee$ lie on a smooth $m$-submanifold $\sss$ of $\R^d$
with reach $\rch(\sss) >16\e$, then $\ripsse(\eee)$ has at most
$2^{35^m}|\eee|$ simplices.
\end{lemma}
\begin{proof}
Given an arbitrary point $v\in\eee$, we will show that the number of
vertices in the star of $v$ in $\ripsse(\eee)$ is at most $33^d$. From
this follows that the number of simplices in the star of $v$ is
bounded by $2^{33^d}$, which proves the first part of the lemma. Let
$\Lambda$ be the set of vertices in the star of $v$. These vertices
lie within Euclidean distance $16\e$ of $v$, and at least $\e$ away
from one another. It follows that they are centers of
pairwise-disjoint Euclidean $d$-balls of same radius $\frac{\e}{2}$,
included in the $d$-ball of center $v$ and radius
$(16+\frac{1}{2})\e$. Therefore, their number is bounded by
$\frac{\volume B(v, (16+\nicefrac{1}{2})\e)}{\volume
  B(v,\nicefrac{\e}{2})} =
\left(\frac{16+\nicefrac{1}{2}}{\nicefrac{1}{2}}\right)^d=33^d$.

Assume now that $v$ and the points of $\Lambda$ lie on a smooth
$m$-submanifold $\sss$ of $\R^d$, such that $16\e<\rch(\sss)$. It
follows then from Lemma 6 of \cite{gw-sdim-04} that, for all
$u\in\Lambda$, we have $\|u-u'\|\leq
\frac{\|u-v\|^2}{2\rch(\sss)}\leq\frac{\e^2}{2\rch(\sss)}<\frac{\e}{32}$,
where $u'$ is the orthogonal projection of $u$ onto the tangent space
of $\sss$ at $v$, $\tangentS(v)$. As a consequence, the orthogonal
projections of the points of $\Lambda$ onto $\tangentS(v)$ lie at
least $\frac{31\e}{32}$ away from one another, and still at most
$16\e$ away from $v$. As a result, they are centers of
pairwise-disjoint open $m$-balls of same radius $\frac{31\e}{64}$,
included in the open $m$-ball of center $v$ and radius
$\left(16+\frac{31}{64}\right)\e$ inside $\tangentS(v)$. Therefore,
their number is bounded by
$\left(\frac{16+\nicefrac{31}{64}}{\nicefrac{31}{64}}\right)^m\leq
35^m$, which proves the second part of the lemma, by the same argument
as above.
\end{proof}

In cases where the input point cloud $\wit$ lies on a smooth
$m$-submanifold $\sss$ of $\R^d$, the above result\footnote{Note that,
  at every iteration $i$ of the algorithm, $\eeei$ is an $\ei$-sparse
  point set, since the algorithm always inserts in $\eee$
  the point of $\wit$ lying furthest away from $\eee$ --- see {\em
    e.g.}  \cite[Lemma~4.1]{go-ruwc-07}.}  suggests that the course of
the algorithm goes through two phases: first, a transition phase, in
which the landmark set $\eee$ is too coarse for the dimensionality of
$\sss$ to have an influence on the shapes and sizes of the stars of
the vertices of $\ripsse(\eee)$; second, a stable phase, in which the
landmark set is dense enough for the dimensionality of $\sss$ to play
a role. This fact is quite intuitive: imagine $\sss$ to be a simple
closed curve, embedded in $\R^d$ in such a way that it roughly fills
in the space within the unit $d$-ball. Then, for large values of $\e$,
the landmark set $\eee$ is nothing but a sampling of the $d$-ball, and
therefore the stars of its points in $\ripsse(\eee)$ are
$d$-dimensional.

Let $\io$ be the last iteration of the transition phase, {\em i.e.}
the last iteration such that $\eio\geq\frac{1}{16}\;\rch(\sss)$. Then,
Lemmas \ref{lem-one-iteration} and \ref{lem-size} imply that the 
time complexity of the transition phase 
is $O(|\wit||\eeeio|^2+8^{33^d}|\eeeio|^5)$, while 
the one of of the stable phase 
is
$O(8^{35^m}\;|\wit|^5)$. We can get rid of the terms
depending on $d$ in at least two ways:

$\bullet$ The first approach has a rather theoretical flavor: it
consists in amortizing the cost of the transition phase by assuming
that $\wit$ is sufficiently large. Specifically, since $\eeeio$ is an
$\eio$-sparse sample of $\sss$, with
$\eio\geq\frac{1}{16}\;\rch(\sss)$, the size of $\eeeio$ is bounded
from above by some quantity $c_0(\sss)$ that depends solely on the
(smooth) manifold $\sss$ -- see {\em e.g.}  \cite{geometrica-5064a}
for a proof in the special case of smooth surfaces. As a result, we
have $8^{33^d}|\eeeio|^k\leq 8^{35^m}|\wit|^k$ for all $k\geq 1$
whenever $|\wit|\geq 8^{33^d-35^m}\;c_0(\sss)$. This condition on the
size of $\wit$ translates into a condition on $\delta$, by a similar
argument to the one invoked above.

$\bullet$ The second approach has a more algorithmic flavor, and it is
based on a backtracking strategy. 
Specifically, we first run the algorithm
without maintaining $\ripsqe(\eee)$ and $\ripsse(\eee)$, which simply
sorts the points of $\wit$ according to their order of insertion in
$\eee$.
Then, we run the algorithm backwards, starting with
$\eee=\eee(|\wit|)=\wit$ and considering at each iteration $j$ the
landmark set $\eee(|\wit|-j)$. During this second phase, we do
maintain $\ripsqe(\eee)$ and $\ripsse(\eee)$ and compute their
persistent Betti numbers. 
If $\wit$ samples $\sss$ densely enough, then Theorem
\ref{th-guarantees} ensures that the relevant plateaus will be
computed before the transition phase starts, and thus before the size
of the data structure becomes independent of the dimension of
$\sss$. It is then up to the user to stop the process
 when the space complexity becomes too large. 

In both cases, we get the following complexity bounds:
\begin{theorem}\label{th-complexity}
If $\wit$ is a point cloud in Euclidean space $\R^d$, then the 
running time of the algorithm 
is $O(8^{33^d}|\wit|^5)$, where $|\wit|$ denotes
the size of $\wit$. If in addition $\wit$ is a $\delta$-sample of some
smooth $m$-submanifold of $\R^d$, with $\delta$ small enough, then the
running time becomes
 $O(8^{35^m}\;|\wit|^5)$.
\end{theorem}
%

%

%
\section{Conclusion}
\label{sec-conclusion}
This paper makes effective the approach developped in
\cite{cl-sctis-07,ceh-spd-05} by providing an efficient, provably good
and easy-to-implement algorithm for topological estimation of general
shapes in any dimensions. Our theoretical framework can also be used
for the analysis of other persistence-based methods.  Addressing a
weaker version of the classical reconstruction problem, we introduce
an algorithm that ultimately outputs a nested pair of complexes at a
user-defined scale, from which the homology of the underlying shape
$X$ are inferred.  When $X$ is a smooth submanifold of $\R^d$, the
complexity scales up with the intrinsic dimension of $X$. These
results provide a new step towards reconstructing (low-dimensional)
manifolds in high-dimensional spaces in reasonnable time with
topological guarantees. It is now tempting to tackle the more
challenging problem of constructing an embedded simplicial complex
that is topologically and geometrically close to the sampled shape. As
a first step, we intend to adapt our method to provide a single output
complex that has the same homology as $X$, using for instance the {\em
  sealing} technique of \cite{fc-mlhc-07}.


\bibliography{geom,geometrica} \bibliographystyle{plain}

\newpage
\tableofcontents


\end{document}